%% file: main.tex
\documentclass[journal,onecolumn, draftclsnofoot, 12pt]{IEEEtran}

\usepackage[utf8]{inputenc} 
\usepackage[T1]{fontenc}
\usepackage{cite}
\usepackage[cmex10]{amsmath} 
                             
\usepackage{pgfplots}
\DeclareUnicodeCharacter{2212}{−}
\usepgfplotslibrary{groupplots,dateplot}
\usetikzlibrary{patterns,shapes.arrows}
\pgfplotsset{compat=1.8}

\usepackage{amsfonts}
\usepackage{amssymb}
\usepackage{mathtools}
\usepackage{enumitem}
\usepackage{tikz}
\usepackage{pgfgantt}
\usepackage{prettyref}
\usepackage{refstyle}
\usepackage{amsthm}
\usepackage{xspace}
\usepackage{comment}
\pgfplotsset{compat=1.3} 
\pgfplotsset{plot coordinates/math parser=false, width=1\columnwidth,     
height=0.8\columnwidth}
\usepackage[caption=false, font=normalsize, labelfont=sf, textfont=sf]{subfig}

\theoremstyle{definition}
\newtheorem{example}{\protect\examplename}
\theoremstyle{plain}
\newtheorem{rem}{\protect\remarkname}
\theoremstyle{definition}
\newtheorem{defn}{\protect\definitionname}
\theoremstyle{plain}
\newtheorem{thm}{\protect\theoremname}
\theoremstyle{plain}
\newtheorem{lem}{\protect\lemmaname}
\ifx\proof\undefined\
  \newenvironment{proof}[1][\proofname]{\par
    \normalfont\topsep6\p@\@plus6\p@\relax
    \trivlist
    \itemindent\parindent
    \item[\hskip\labelsep
          \scshape
      #1]\ignorespaces
  }{%
    \endtrivlist\@endpefalse
  }
  \providecommand{\proofname}{Proof}
\fi

\newref{lem}{name=Lemma~}
\newref{thm}{name=Theorem~}
\newref{exa}{name=Example~}
\newref{sec}{name=Section~}
\newref{def}{name=Definition~}
\newref{fig}{name=Fig.~}
\newref{subsec}{name=Subsection~}
\newref{rem}{name=Remark~}
\newref{eq}{refcmd=(\ref{#1})}
\newref{sce}{name=scenario~}

\providecommand{\definitionname}{Definition}
\providecommand{\examplename}{Example}
\providecommand{\lemmaname}{Lemma}
\providecommand{\remarkname}{Remark}
\providecommand{\theoremname}{Theorem}

\providecommand{\revisioncolor}{black}

\providecommand{\arxiv}{\hspace{-6pt}}
\newcommand{\revision}[1]{\color{\revisioncolor} #1 \color{black}}

\newlength\fwidth

\newcommand{\proposed}{SBP}

\begin{document}

\bstctlcite{IEEEexample:BSTcontrol}
\title{Bivariate Polynomial Codes for Secure Distributed Matrix Multiplication} 

\author{Burak Hasırcıoğlu,
Jesús Gómez-Vilardebó,
and~Deniz Gündüz,
\thanks{Burak Hasırcıoğlu and Deniz Gündüz are with the Department
of Electrical and Electronic Engineering, Imperial College London, UK. E-mail: \{b.hasircioglu18, d.gunduz\}@imperial.ac.uk}
\thanks{Jesús Gómez-Vilardebó is with Centre Tecnològic de Telecomunicacions
de Catalunya (CTTC/CERCA), Barcelona, Spain. E-mail: jesus.gomez@cttc.es}
\thanks{A preliminary version of this paper has been accepted for a presentation in 2021 IEEE International Symposium on Information Theory (ISIT) \cite{hasircioglu2021speeding}. }}

\maketitle

\begin{abstract}
We consider the problem of secure distributed matrix multiplication (SDMM). Coded computation has been shown to be an effective solution in distributed matrix multiplication, both providing privacy against workers and boosting the computation speed by efficiently mitigating stragglers. 
In this work, we present a non-direct secure extension of the recently introduced bivariate polynomial codes. Bivariate polynomial codes have been shown to be able to further speed up distributed matrix multiplication by exploiting the partial work done by the stragglers rather than completely ignoring them while reducing the upload communication cost and/or the workers' storage's capacity needs. 
We show that, especially for upload communication or storage constrained settings, the proposed approach reduces the average computation time of SDMM compared to its competitors in the literature.
\end{abstract}


\vspace{-7pt}
\section{Introduction}\label{sec:intro}

Matrix multiplication is a fundamental building block of many applications in signal processing and machine learning. For some applications, especially those involving massive matrices and 
stringent latency requirements, matrix multiplication in a single computer is infeasible,
and distributed solutions need to be adopted. In such scenarios, the full multiplication task is first partitioned into smaller sub-tasks, which are then distributed across dedicated \emph{workers}. 

In this work, we address two main challenges in distributed matrix multiplication. The first one is referred to as the \emph{stragglers} problem, which refers to unresponsive or slow workers. If completing the full task requires the completion of the computations assigned to all the workers, then straggling workers become a significant bottleneck. To avoid stragglers, additional redundant computations can be assigned to workers. It has been recently shown that the use of error-correcting codes, by treating the slowest workers as erasures instead of simply replicating tasks across workers, significantly lowers the overall computation time \cite{lee2017speeding}. In the context of straggler mitigation, polynomial-type codes are studied in \cite{yu_polynomial_2017, dutta_optimal_2019, yu_straggler_2018-1, jia2019cross}. In these schemes, matrices are first partitioned and encoded using polynomial codes at the master server. Then, workers compute sub-products by multiplying these coded partitions and send the results back to the master for decoding. The minimum number of sub-tasks required to decode the result is referred to as the \emph{recovery threshold} and denoted by $R_{th}$. 
All these works assume that only one sub-product is assigned to each worker, and therefore, any work done by the workers beyond the fastest $R_{th}$ is completely ignored. This is sub-optimal, particularly when the workers have similar computational speeds.
This problem is addressed by the \emph{multi-message approaches} in \cite{kiani_exploitation_2018, amiri_computation_2018, ozfatura2020straggler, hasircioglu2020bivariate}. In these works, multiple sub-products are assigned to each worker, and the result of each sub-product is communicated to the master as soon as it is completed. This results in faster completion of the full computation as it allows to exploit partial computations completed by stragglers. Moreover, the multi-message approach makes finishing the task possible even if there are not as many available workers as the recovery threshold.
However, as discussed in \cite{hasircioglu2020bivariate}, a direct extension of polynomial-type codes to the multi-message setting by simply assigning multiple sub-products to the workers increase the \emph{upload communication costs}, which is defined as the number of bits sent from the master to each worker, or equivalently, the storage required per worker. 
To combat this effect, product codes are proposed in \cite{kiani_exploitation_2018} for the multi-message distributed matrix multiplication problem. However, with product codes, every sub-product is not equally useful while decoding the full-product, i.e., they are not one-to-any replaceable, which degrades their performance.
The \textit{bivariate polynomial codes} are introduced in \cite{hasircioglu2020bivariate} to address this issue, achieving a better trade-off between the upload cost and average computation time. 

The second challenge we tackle in this paper is privacy. The multiplied matrices may contain sensitive information, and even partially sharing these matrices with the workers may cause a privacy breach. Moreover, several workers can exchange information with each other to learn about the multiplied matrices. Such a collusion may result in a leakage even if no information is revealed to individual workers. The first application of polynomial codes to privacy-preserving distributed matrix multiplication is presented in \cite{chang2018capacity}. To hide the matrices from the workers, random matrix partitions are created, and linearly encoded together with the true matrix partitions using polynomial codes. This requires increasing the degree of the encoding polynomial and thus increasing the recovery threshold. The recovery threshold has been improved in subsequent works \cite{kakar2018rate}, \cite{d2020gasp}, by carefully choosing the degrees of the encoding monomials so that the resultant decoding polynomial contains the minimum number of additional coefficients. In \cite{aliasgari2020private, jia2021capacity,kakar2019uplink}, lower recovery threshold values than \cite{d2020gasp} are obtained by using different matrix partitioning techniques and different choices of encoding polynomials, at the expense of a considerable increase in the upload cost. In \cite{mital2020secure}, a novel coding approach for distributed matrix multiplication is proposed based on polynomial evaluation at the roots of unity in a finite field. It has constant time decoding complexity and a lower recovery threshold than traditional polynomial-type coding approaches, but the sub-tasks are not one-to-any replaceable and its straggler mitigation capability is limited. In \cite{bitar2020rateless}, a multi-message approach is proposed for SDMM by using rateless codes. Computations are assigned adaptively in rounds, and in each round, workers are classified into clusters depending on their computation speeds. Results from a worker in a cluster are useful for decoding only if the results of all the sub-tasks assigned to that cluster and also to the fastest cluster are collected, making computations not one-to-any replaceable. Still, the strategy exhibited good average computation times by estimating and adapting to the computation speeds of the workers.


In this work, we propose Secure Bivariate Polynomial (SBP) codes, for the multi-message, straggler-resistant, SDMM task based on bivariate Hermitian polynomial codes.
We show that under a limited upload cost budget, or when the number of fast workers is limited, SBP codes outperform other schemes in the literature in terms of the average computation time. We also show that this scheme retains its low average computation time when the computation speeds of the workers significantly differ, i.e., heterogeneous scenario, or when they are close to each other, i.e., homogeneous scenario. In addition, we propose an extension of GASP codes \cite{d2020gasp} to the multi-message setting and evaluate its performance. We show that when the upload cost budget is sufficiently high, the proposed extension could considerably lower the average computation time of the SDMM task.

\vspace{-7pt}
\section{Problem Setting\label{sec:Problem-Setting}}

We study distributed matrix multiplication with strict privacy
requirements. The elements of our matrices are in a finite field $\mathbb{F}$,
and we denote the size of the finite field by $q$. The \textit{master} wants to multiply statistically independent matrices $A\in\mathbb{F}^{r\times s}$ and $B\in\mathbb{F}^{s\times t}$, $r,s,t\in\mathbb{Z}^{+}$, with the help of $N$ dedicated \textit{workers}, which possibly have heterogeneous computation speeds and storage capacities.

To offload the computation, the master divides the multiplication task into smaller sub-tasks, which are then assigned to workers. The master partitions $A$ into $K$ sub-matrices as $A=\begin{bmatrix}A_{1}^{T} & A_{2}^{T} & \cdots & A_{K}^{T}\end{bmatrix}^{T}$,
where $A_{i}\in\mathbb{F}^{\frac{r}{K}\times s}$, $\forall i\in[K]\triangleq\{1,2,\dots,K\}$,
and $B$ into $L$ sub-matrices as $B=\begin{bmatrix}B_{1} & B_{2} & \cdots & B_{L}\end{bmatrix}$,
where $B_{j}\in\mathbb{F}^{s\times\frac{t}{L}}$, $\forall j\in[L]$. The master sends coded versions, i.e., linear combinations, of these partitions to the workers. We assume that there is an \textit{upload cost} constraint per worker, denoted by $u_i$ for worker $i$, which limits the maximum number of bits that can be transmitted from the master to each worker. This upload cost is a limiting factor on the number of coded partitions of $A$, denoted by $m_{A,i}$,  and of $B$, denoted by $m_{B,i}$, that can be sent to each worker.  More specifically, for worker $i$, $m_{A,i}$ and $m_{B,i}$  must satisfy $\left(m_{A,i}rs/K + m_{B,i}st/L\right)\log_2(q) \leq u_i$. Provided that they comply with the upload cost constraint,  $m_{A,i}$ and $m_{B,i}$ are chosen depending on the underlying coding scheme and the master sends coded partitions $\tilde{A}_{i,k}$ and $\tilde{B}_{i,l}$ to worker $i$, where $i\in[N]$, $k\in[m_{A,i}]$ and $l\in [m_{B,i}]$. For simplicity, we describe a static setting, in which all the coded matrices are sent to the workers before they start computations. In a more dynamic scenario, matrix partitions can be delivered when they are needed, which would reduce the memory required at the workers. 
The workers multiply the received coded partitions of $A$ and $B$ as instructed by the underlying coding scheme and send the result of each computation to the master as soon as it is completed. Once the master receives a number of computations equal to the \textit{recovery threshold}, $R_{th}$, it can decode the desired multiplication $AB$.

In our threat model, the workers are honest but curious. They follow the protocol, but they can use the received encoded matrices to obtain information about $A$ and $B$. We assume that any $T$ workers can collude, i.e., exchange information among themselves. Our privacy requirement is that no $T$ workers are allowed to gain any information about the content of the multiplied matrices. That is,
\color{\revisioncolor}$$I\left(A,B;\{\tilde{A}_{i,k},\tilde{B}_{i,l}\mid i\in \mathcal{N}, k\in[m_{A,i}],l\in[m_{B,i}]\}\right)=0,$$ \color{black} where $I$  is the mutual information \color{\revisioncolor}and $\mathcal{N}$ is the any subset of $[N]$ with cardinality at most $T$. \color{black}

Under this setting, the main problem we attempt to solve is minimizing the \emph{average computation time}, which is defined as the time required for the master to collect sufficiently many computations to decode the desired computation $AB$. We assume that the workers' computation speeds can be homogeneous, i.e., the average speed of each available worker is close to each other, or heterogeneous, in which the average speeds of the workers vary. Workers can also straggle, i.e., become unresponsive temporarily. 

\color{\revisioncolor}
\section{Extension of Bivariate Polynomial Codes for Secure DMM}\label{sec:naive_extension}

As a first attempt to improve the upload cost efficiency of SDMM, we provide the naive extension of bivariate polynomial codes proposed in \cite{hasircioglu2020bivariate} to SDMM. In \cite{hasircioglu2020bivariate}, the partitioning of the matrices is as described in \secref{Problem-Setting} and the two encoding polynomials are generated as
\begin{equation}
A(x)=A_{1}+A_{2}x+\cdots+A_{K}x^{K-1},
\end{equation}
\begin{equation}
B(y)=B_{1}+B_{2}y+\cdots+B_{L}y^{L-1}.
\end{equation}
Therefore, at the master, the goal is to interpolate the following polynomial.
\begin{equation}
A(x)B(y)=\sum_{i=1}^{K}\sum_{j=1}^{L}A_{i}B_{j}x^{i-1}y^{j-1}.
\end{equation}
Since worker $i\in[N]$ can store $m_{A,i}$ partitions of $A$ and $m_{B,i}$ partitions of $B$, the master sends the first $m_{A,i}$ derivatives of $A(x)$ and first $m_{B,i}$ derivatives of $B(y)$, evaluated at $x_i$ and $y_i$, respectively, which are evaluation points of the encoding polynomials chosen distinct for each worker. Each worker multiplies the received encoded partitions of $A(x)$ and $B(y)$ following a specific order from the smaller-order derivatives to larger-order derivatives and sends the results of each computation as soon as it is finished. Then, once the master receives $KL$ computations from the workers, it instructs all the workers to stop and starts decoding $A(x)B(y)$. 

In order to provide a simple direct extension of this scheme to SDMM in which $T$ worker collude, we limit the analysis to the case $m_{A,i}=m_A$ and $m_{B,i}=m_B$, $\forall i \in [N]$. Thus, from a security point of view, each worker gets $m_A$ and $m_B$ coded partitions of $A$ and $B$, respectively. Since up to $T$ workers collude, in total, $m_AT$ coded partitions of $A$ and $m_BT$ coded partitions of $B$ are leaked to the workers. To protect such a leakage, we need to add $m_AT$ and $m_BT$ random matrix partitions to $A(x)$ and $B(y)$, respectively. Thus, the encoding polynomials for this naive extension of bivariate polynomial codes to SDMM becomes
\begin{equation}
    A(x)=A_{1}+A_{2}x+\cdots+A_{K}x^{K-1}+\sum_{i=1}^{m_AT}R_ix^{K+i-1},
\end{equation}
\begin{equation}
B(y)=B_{1}+B_{2}y+\cdots+B_{L}y^{L-1}+\sum_{i=1}^{m_BT}S_ix^{L+i-1},
\end{equation}
where $R_i$ and $S_i$ are matrix partitions chosen uniformly at random from the elements of $\mathbb{F}_q$. Therefore, the polynomial to be interpolated at the master becomes
\begin{multline}
A(x)B(y)=\sum_{i=1}^{K}\sum_{j=1}^{L}A_{i}B_{j}x^{i-1}y^{j-1} + \sum_{i=1}^K\sum_{j=L+1}^{m_BT}A_iS_jx^{i-1}y^{j-1}
+\sum_{i=K+1}^{m_AT}\sum_{j=1}^{L}R_iB_jx^{i-1}y^{j-1}\\ + \sum_{i=K+1}^{m_AT}\sum_{j=L+1}^{m_BT}R_iS_jx^{i-1}y^{j-1}.\label{eq:sdmm_axby}
\end{multline}
Therefore, considering the number of monomials of $A(x)B(y)$ in \eqref{sdmm_axby}, $R_{th}=(K+m_AT)(L+m_BT)$ evaluations of $A(x)B(y)$ are needed to interpolate it, which is $\mathcal{O}(T^2)$ has a quadratic dependence on $T$. 

Observe that in this naive extension, for a worker to provide $m=m_Am_B$ computations, uploading $m_A$ coded partitions of $A$ and $m_B$ coded partitions of $B$ are enough. This means that the upload cost of the scheme is on the order of $\sqrt{m}$. However, the price we pay for such a reduced upload cost is a quadratic dependence of the recovery threshold on the number of colluding workers. Such dependence on $T$ may quickly become restrictive for typical $T$ values and hence, the benefits of the naive extension of bivariate polynomial codes to SDMM may be out-weighted by its drawbacks. Thus, we need more sophisticated schemes that can keep this low upload cost with a better scaling behaviour for the recovery threshold with respect to $m$ and $T$. In the next section, we present our proposed solution for such a problem.

\color{black}
\section{Secure Bivariate Polynomial (SBP) Codes}
Our coding scheme is based on bivariate polynomial codes \cite{hasircioglu2020bivariate}. Compared to their univariate counterparts
, bivariate polynomial codes allow workers to complete more sub-tasks 
under the same upload cost budget, which improves the average computation time and helps to satisfy the privacy requirements.

\vspace{-5pt}
\subsection{Encoding}
\label{subsec:encoding}
In SBP coding scheme, coded matrices are generated by evaluating the following polynomials and their derivatives: 
\begin{equation}
\label{eq:a_x}
A(x)=\sum_{i=1}^{K}A_{i}x^{i-1}+\sum_{i=1}^{T}R_{i}x^{K+i-1},    
\end{equation}
\begin{equation}
\label{eq:b_xy}
B(x,y)=\sum_{i=1}^{L}B_{i}y^{i-1}+\sum_{i=1}^{T}\sum_{j=1}^{m}S_{i,j}x^{K+i-1}y^{j-1},
\end{equation}
where $m \leq L$ is the maximum number of sub-tasks any worker can complete. Matrices $R_i\in \mathbb{F}_q^{\frac{r}{K}\times s}$ and $S_{i,j}\in \mathbb{F}_q^{s \times \frac{t}{L}}$ are independent and uniform randomly generated from their corresponding domain for $i\in [T]$ and $j \in [m]$. For each worker $i$, the master evaluates $A(x)$ at $x_{i}$ and the derivatives of $B(x,y)$ with respect to $y$ up to the
order \color{\revisioncolor}$m-1$ \color{black} at $(x_{i},y_{i})$. We only require these evaluation points to be distinct. Thus, the master sends to worker
$i$, $A(x_{i})$ and $\mathcal{B}_{i}=\{B(x_{i},y_{i}),\partial_{1}B(x_{i},y_{i}),\dots,\partial_{m-1}B(x_{i},y_{i})\}$,
where $\partial_{i}$ denotes the $i^{th}$ partial derivative with respect
to $y$. Thus, we require $m_{A,i}=1$ and $m_{B,i}=m$.

In \eqref{a_x} and \eqref{b_xy}, the role of $R_i$'s and $S_{i,j}$'s is to mask the actual matrix partitions for privacy. The following theorem states that the evaluations of $A(x)$, $B(x,y)$ and its derivatives do not leak any information about $A$ and $B$ to any $T$ colluding workers.

\begin{thm}
For the encoding scheme described above, we have
$ I(A,B;\{A(x_i),\mathcal{B}_i : i\in \mathcal{N}\})=0,$
$\forall \mathcal{N}\subset [N]$ such that \color{\revisioncolor} $|\mathcal{N}|\leq T$. \color{black}
\end{thm}

\begin{proof}
Since $A$ and $B$ are independent, we have
\begin{multline} \label{eq:thm1_proof_mutual_inf}
I(A,B;\{A(x_{i}),\mathcal{B}_{i}: i\in \mathcal{N}\})=
I(A;\{A(x_{i}): i\in \mathcal{N}\}) +I(B;\{\mathcal{B}_{i}: i\in \mathcal{N}\}).
\end{multline}
Let us first bound $I(A;\{A(x_i)|i\in \mathcal{N}\})$ as follows.
\begin{align}
&I\left(A;\{A(x_i):i\in \mathcal{N}\}\right)\nonumber\\
&=H\left(\{A(x_i):i\in \mathcal{N}\}\right)-H\left(\{A(x_i):i\in \mathcal{N}\}|A\right)\\
&=H\left(\{A(x_i):i\in \mathcal{N}\}\right)-H\left(\{R_i:i\in [T]\}|A\right)\\
&\stackrel{(a)}{=}H\left(\{A(x_i):i\in \mathcal{N}\}\right)-T\frac{rs}{K}\log(q)\\
&\stackrel{(b)}{\leq} \sum_{i=1}^{|\mathcal{N}|} H(A(x_i))-T\frac{rs}{K}\log(q)\\
&=|\mathcal{N}|\frac{rs}{K}\log(q)-T\frac{rs}{K}\log(q) \stackrel{(c)}{\leq} 0, \label{eq:mi_first_component}
\end{align}
where (a) follows from the fact that $R_i$'s are independent from each other and from $A$, (b) is due to the fact that joint entropy of several random variables is upper bounded by the sum of the individual entropies of these random variables and (c) is due to $\mathcal{N}\leq T$.

We can bound $I(B;\{\mathcal{B}_i:i\in \mathcal{N}\})$ similarly as follows.
\begin{align}
&I\left(B;\{\mathcal{B}_i:i\in \mathcal{N}\}\right)\nonumber\\
&=H\left(\{\mathcal{B}_i:i\in \mathcal{N}\}\right)-H\left(\{\mathcal{B}_i:i\in \mathcal{N}\}|B\right)\\
&=H\left(\{\mathcal{B}_i:i\in \mathcal{N}\}\right)-H(\{S_{i,j}:i\in[T],j\in[m]\}|B)\\
&=H\left(\{\mathcal{B}_i:i\in \mathcal{N}\}\right)-Tm\frac{st}{L}\log(q)\\
&\leq \sum_{i=1}^{|\mathcal{N}|} \sum_{j=1}^m H(B(x_i,y_j))-Tm\frac{st}{L}\log(q)\\
&=|\mathcal{N}|m\frac{st}{L}\log(q)-Tm\frac{st}{L}\log(q) \leq 0. \label{eq:mi_second_component}
\end{align}
The claim follows by substituting \eqref{mi_first_component} and \eqref{mi_second_component} into \eqref{thm1_proof_mutual_inf}.
\vspace{-5pt}
\end{proof}
\vspace{-5pt}
\subsection{Computation}\label{subsec:Computation}

Worker $i$ multiplies $A(x_{i})$ and $\partial_{j-1} B(x_{i},y_{i})$ with the increasing order of $j\in[m]$. That is, $j^{th}$ completed computation is $A(x_{i})\partial_{j-1} B(x_{i},y_{i})$. As soon as each multiplication
is completed, its result is communicated back to the master.
\vspace{-5pt}
\subsection{Decoding}


After collecting sufficiently many computations from the workers, the master can interpolate $A(x)B(x,y)$. Note that, in our scheme, every computation is equally useful; that is, the sub-tasks are one-to-any replaceable. In the following theorem, we give the recovery threshold expression, which specifies the minimum number of required computations and characterizes the probability of decoding failure, i.e., bivariate polynomial interpolation, due to the use of a finite field. 

\begin{thm} \label{thm:r_th}
Assume the evaluation points $(x_i,y_i)$ are chosen uniform
randomly over the elements of $\mathbb{F}$. If the number of computations of sub-tasks received from the workers, which obey
the computation order specified in \subsecref{Computation} is greater than the recovery threshold $R_{th} \triangleq (K+T)L+m(K+T-1)$, then with probability at least $1-d/q$,
the master can interpolate the unique polynomial $A(x)B(x,y)$,
where
\begin{multline}
\label{eq:thm2}
d\triangleq \frac{m}{2}\left(3(K+T)^{2}+m(K+T)-8K-6T-m+3\right)
+\frac{(K+T)L}{2}\left(K+L+T-2\right).
\end{multline}

\end{thm}

We give the proof sketch of \thmref{r_th} in \secref{main_proof}. \thmref{r_th} states that we can make the probability of failure arbitrarily small by increasing the order $q$ of the finite field.

\begin{thm}\label{thm:upload_cost}
The total upload cost of the SBP coding scheme is $N\left(rs/K+mst/L\right)\log_2(q)$ bits.
\end{thm}

\begin{proof}
The SBP coding scheme assigns
$m$ computations to each worker,
by sending 
one coded partition of $A$ and $m$ coded partitions of $B$. Remember that each coded partition of $A$ is a matrix of size $\frac{r}{K}\times s$ and each coded partition of $B$ is a matrix of size $s \times \frac{t}{L}$. Since there are $N$ workers, the master uploads $N\left(rs/K+mst/L\right)$ elements of the field $\mathbb{F}$. Since, in total, there are $q$ elements in $\mathbb{F}$, the total upload cost is $N\left(rs/K+mst/L\right)\log_2(q)$ bits. 
\end{proof}


\begin{rem}\label{rem:model-independ}

The SBP scheme does not exploit any  parameter of the underlying statistical model of the workers' speeds. Under a total upload cost constraint, if no prior information about the computation speeds of the workers is available, then assigning more computation load, $m$, to every worker is a favorable approach. Although this increases the recovery threshold, i.e., the term $m(K+T-1)$, the faster workers do not run out of computations easily, avoiding the slowest workers dominating the computation time. The benefit of this prevails over the detriment due to the increase in the recovery threshold. Surely, if prior information about the computation speeds of the workers is available, we could exploit it by assigning more, but still less than $L$, computations to faster workers, which would result in fewer number of coded partitions leaked to the colluding workers. In this case, the recovery threshold would be lower, further increasing the protection against stragglers. However, the SBP scheme has been designed as agnostic to the delay model of the workers and  specifically to maximize the number of sub-tasks delivered by a worker under an upload cost constraint.
Thanks to the extra computations at workers, we show in our simulation results that a model-independent version of SBP is enough to beat model-dependent schemes such as the one in \cite{bitar2020rateless}. Thus, we expect the SBP scheme to work for large varieties of statistical models of the worker's speeds. 
\end{rem}

\vspace{-5pt}
\section{Extension of GASP codes to multi-message setting}
\label{sec:discussion}

State of the art schemes in SDMM \cite{chang2018capacity, kakar2018rate} are combined and improved in \cite{d2020gasp}, referred to as GASP codes. Originally, GASP codes are designed for the single-message scenario, in which each worker is assigned a single computation task. In this section, we extend the GASP codes to the multi-message setting, which we call multi-message GASP (MM-GASP) scheme. The encoding polynomials for the GASP codes are $$A(x)=\sum_{i=1}^{K}A_i x^{\alpha_i}+\sum_{i=1}^{T}R_i x^{\alpha_{K+i}},$$
$$B(x)=\sum_{i=1}^{L}B_i x^{\beta_i}+\sum_{i=1}^{T}S_i x^{\beta_{L+i}},$$ where $R_i$'s and $S_i$'s are random matrix partitions, and $\alpha_i$'s and $\beta_i$'s are determined such that \color{\revisioncolor} $A_iB_j$, $\forall i \in[K], \forall j\in[L]$ can be decoded and \color{black} the number of monomials whose coefficients consist of undesired terms, such as multiplications involving $R_i$'s and $S_i$'s, in $A(x)B(x)$, are minimized. We do not cover the details of how $\alpha_i$'s and $\beta_i$'s are determined, but as a result of the process detailed in \cite{d2020gasp}, the recovery threshold becomes $R_{th}^{GASP}\color{\revisioncolor}(K,L,T)\color{black}=$
\begin{equation}
\label{eq:r_th_gasp}
\begin{cases}
\begin{array}{lcc}
KL+K+L,&  & 1=T<L\leq K\\
KL+K+L+T^{2}+T-3,&  & 1<T<L\leq K\\
\left(K+T\right)(L+1)-1,&  & L\leq T<K\\
2KL+2T-1,&  & L\leq K\leq T.
\end{array}\end{cases}    
\end{equation}

For the extension of GASP codes to the multi-message setting, i.e., MM-GASP, we assign $m>1$ tasks to each worker. Thus, a worker can see $m$ evaluations of $A(x)$ and $B(x)$, and any $T$ colluding workers can see $mT$ evaluations. 
Thus, to make the scheme secure against $T$ colluding workers, we need to add $mT$ random matrix partitions to each encoding polynomial, instead of $T$. Thus, we have
$$A(x)=\sum_{i=1}^{K}A_i x^{\alpha_i}+\sum_{i=1}^{mT}R_i x^{\alpha_{K+i}},$$
$$B(x)=\sum_{i=1}^{L}B_i x^{\beta_i}+\sum_{i=1}^{mT}S_i x^{\beta_{L+i}}.$$

\begin{thm}
\label{thm:r_th_gasp_mm}
The recovery threshold of MM-GASP is given by $R_{th}^{MM-GASP}(K,L,T)=$
\begin{equation}
\label{eq:r_th_gasp_mm}
\begin{cases}
\begin{array}{lcc}
KL+K+L,&  & 1=mT<L\leq K\\
KL+K+L+(mT)^{2}+mT-3,&  & 1<mT<L\leq K\\
\left(K+mT\right)(L+1)-1,&  & L\leq mT<K\\
2KL+2mT-1,&  & L\leq K\leq mT.
\end{array}\end{cases}
\end{equation}
\end{thm}

\begin{proof}
We can define $\tilde{T}= mT$ and write
\begin{align}
    A(x)=&\sum_{i=1}^{K}A_i x^{\alpha_i}+\sum_{i=1}^{\tilde{T}}R_i x^{\alpha_{K+i}},\\ B(x)=&\sum_{i=1}^{L}B_i x^{\beta_i}+\sum_{i=1}^{\tilde{T}}S_i x^{\beta_{L+i}}.  
\end{align}
Now let us consider 
\begin{multline}
A(x)B(x)=\sum_{i=1}^{K}\sum_{j=1}^{L}A_iB_jx^{\alpha_i+\beta_i} + \sum_{i=1}^{K}\sum_{j=1}^{\tilde{T}}A_iS_jx^{\alpha_i+\beta_{L+i}}
+\sum_{i=1}^{\tilde{T}}\sum_{j=1}^{L}R_iB_jx^{\alpha_{K+i}+\beta_i}\\ + \sum_{i=1}^{\tilde{T}}\sum_{j=1}^{\tilde{T}}R_iS_jx^{\alpha_{K+i}+\beta_{L+i}}.
\end{multline}
According to the proof of \eqref{r_th_gasp} in \cite{d2020gasp}, $\alpha_i$'s and $\beta_i$'s are chosen such that: 1) from the evaluations of $A(x)B(x)$, $A_iB_j$'s $\forall i\in[K], j\in[L]$ are decodable and 2) the number of monomials in $A(x)B(x)$ whose coefficients are undesired terms are minimized. For this, we only need to consider the structure of $A(x)B(x)$ and $m$ itself is not related other than determining the value of $\tilde{T}$. Therefore, the problem reduces to deriving the recovery threshold of a classical GASP coding scheme when $\tilde{T}$ workers collude, which is $R_{th}^{GASP}(K,L,\tilde{T})$ by \eqref{r_th_gasp}. If we substitute $\tilde{T}=mT$, then we obtain \eqref{r_th_gasp_mm}.

\end{proof}


\begin{rem}\label{rem:uploadcostmmGasp}
In multi-message univariate polynomial coding schemes, such as in MM-GASP codes we have just introduced, if a worker is assigned $m$ sub-tasks, then $m$ coded partitions of both $A$ and $B$ are required. Thus, the total upload cost of MM-GASP is $Nm\left(rs/K+st/L\right)\log_2(q)$ bits, which is larger than that of SBP coding scheme.
\end{rem}


The recovery thresholds of the MM-GASP codes and SBP codes can be compared as a function of the number of coded partitions $m$, by direct inspection of the recovery thresholds in \eqref{r_th_gasp_mm} and \thmref{r_th}. Observe that SBP coding scheme's recovery threshold is smaller than that of the MM-GASP code if $L\leq mT<K$, and  $K\leq TL+1$, which is satisfied as $K$ and $L$ become close to each other, or, if $L\leq K\leq mT$, and $(K-T)(L-m)\geq (1-m)$ is satisfied. In \figref{r_th_vs_m}, we provide the recovery thresholds of the two schemes as a function of the number of computations allocated to each worker for $K=L=100$ and $T=30$.

\begin{figure}[t]
    \centering
    \resizebox{!}{0.5\linewidth}{
    \input{figures/r_th_vs_m_t30}}
    \caption{$R_{th}$ vs. the number of computations assigned to each worker for SBP coding scheme and the MM-GASP scheme for $K=L=100$ and $T=30$.}
    \label{fig:r_th_vs_m}
\end{figure}
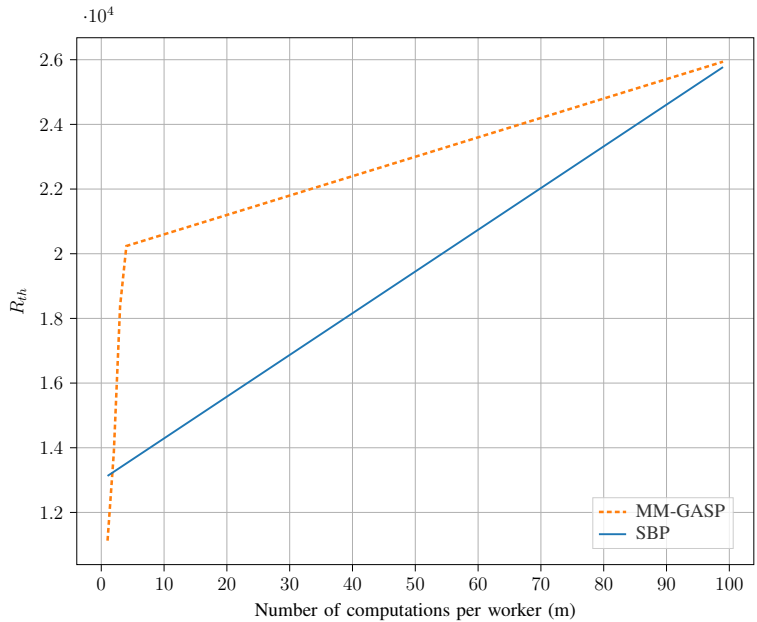

We note that such a comparison may only be meaningful in the unlimited upload cost budget scenario. Otherwise, comparing the recovery thresholds for the same $m$ might be misleading since, for a given upload cost constraint, each scheme provides a different number of sub-tasks, $m$, to workers, as detailed in \thmref{upload_cost} and \remref{uploadcostmmGasp}, for SBP and for MM-GASP, respectively. We provide further discussion on this issue in the next section, see  \figref{r_th}, where we show the recovery thresholds as a function of the total upload cost budget for a scenario with $K=L=100$ and $T=30$. 

\section{Simulation Results and Discussion}
\label{sec:simulations}

In this section, we compare SBP codes with MM-GASP and the rateless coding scheme proposed in \cite{bitar2020rateless} in terms of the trade-off between the average computation time (ACT) and the total upload cost budget (UCB), under the scenarios with heterogeneous and homogeneous workers.

The comparison between the MM-GASP scheme and SBP coding scheme is direct, as both are based on the same set of assumptions. They achieve different recovery thresholds as a function of the $L,M,T$ and $m$, but they both assume that the coded submatrices are uploaded only once before the computations start, no prior knowledge of the computation speeds of workers is needed or can be exploited, and the first $R_{th}$ results received from any subset of the workers allow recovering the desired computation. However, the setting and the assumptions in \cite{bitar2020rateless} are slightly different. In the rateless coding scheme proposed in \cite{bitar2020rateless}, computations are organized in rounds. If the speeds of the workers are not already known, in the first round, every worker is assigned one computation to estimate their speeds. 
Then, based on the known or estimated speeds, workers are grouped into $c$ clusters, such that the workers with similar speeds are in the same cluster. We denote by $n_u$  the number of workers belonging to cluster $u$, $u\in [c]$. In each round, for any computation within a cluster to be useful for decoding, we need that all the workers in that cluster and also all workers in cluster one, which is special, to finish their assigned tasks. Once all the workers in cluster $u$ and cluster $1$ finish their tasks, they provide $d_u$, and $d_1$ useful computations to the master, where $d_1=\left\lfloor (n_1 -2T+1)/2 \right\rfloor$ and $d_u=\left \lfloor (n_u -T+1)/2 \right \rfloor$ for $2\leq u \leq c$. No further synchronization is needed among clusters, and a new task can be assigned to a worker as soon as it finishes its assigned task.  Once $KL(1+\epsilon)$ useful computations are collected by the master from different clusters across multiple rounds, 
the decoding procedure can start. Here, $\epsilon$ is the overhead due to the Fountain codes used in \cite{bitar2020rateless}, which in our simulations is set to 0.05.  
The performance of this scheme
depends critically on how good the distribution of the workers' speeds can be estimated. Observe that, in the event that a worker in a "fast" cluster straggles, the finishing time of the overall cluster can be arbitrarily delayed. 
This is the main drawback of this scheme in comparison with SBP coding scheme and the MM-GASP, for which any computation at any worker is equally useful. 
The clear advantage of the rateless coding scheme is that the computation load $m_i$, i.e., the number of tasks assigned to worker $i$, does not need to be specified in advance, and tasks can be dynamically allocated to workers in each cluster across rounds. Moreover, the recovery threshold is not dominated by the maximum computation load $m=\max m_i$, as is the case for SBP coding and the MM-GASP schemes.  Therefore, in order to allow the rateless coding scheme to benefit from  this flexibility, in our simulations, we consider a total upload cost for \cite{bitar2020rateless}, i.e., the computations are assigned to clusters until the total upload communication budget is met, while for MM-GASP and SBP we impose an upload cost constraint per worker.
We emphasize that this is a relaxation of the problem formulation introduced in \secref{Problem-Setting}, and is only applied to the rateless coding scheme. Although SBP coding scheme and the MM-GASP code can also benefit from this relaxation when the computation statistics of the workers are known, such optimization is out of the scope of this paper and will be considered in future work.

In our simulations, following the literature \cite{liang2014tofec,lee2017speeding}, we assume that the time for a worker to finish one sub-task is distributed as a shifted exponential random variable with density $f(t)=\lambda e^{-\lambda(t-\nu)}$ for $t\geq \nu$, and $f(t)=0$ otherwise, where the scale parameter $\lambda$ controls the speed of the worker and the shift parameter $\nu$ is the minimum time duration for a task to be completed. Smaller $\lambda$ implies slower workers and more tendency to straggle. In each scenario, we run 1000 experiments independently with the given parameters and present the average computation time. We assume that the partitions of matrices $A$ and $B$ have the same size, i.e., $\frac{r}{K}=\frac{t}{L}$, in all of the scenarios considered. Given that the computation time per sub-task is a fraction $\frac{1}{KL}$ of the complete task, to facilitate the comparison between different configurations, we choose $\lambda \propto KL$, and $\nu \propto \frac{1}{KL}$, in all simulation setups.

\vspace{-5pt}
\subsection{Heterogeneous Workers} \label{subsec:heter}

In this subsection, we assume that the computation speeds of the workers are heterogeneous. Specifically, we assume six \emph{heterogeneity classes}, with scale parameters $\lambda_1=10^{-1}\times KL$, $\lambda_2=10^{-1.5}\times KL$, $\lambda_3=10^{-2}\times KL$, $\lambda_4=10^{-2.5}\times KL$, $\lambda_5=10^{-3}\times KL$ and $\lambda_6=10^{-3.5}\times KL$, and a common shift parameter of $\nu=10/(KL)$ seconds. There are 75 workers for each class summing up to $N=450$ workers in total, and assume that any subset of at most $T=N/15$ workers can collude. We divide both matrices $A$ and $B$ into $K=L=100$ partitions.
We evaluate the scheme in \cite{bitar2020rateless} for several numbers of clusters, $c$, to observe the effect of the mismatch between the actual number of heterogeneity classes and the chosen $c$ value. While generating the clusters, we simply assign around $N/c$ workers to each cluster, according to the estimated speeds in the first round. We do not change the parameter $c$ across rounds.

First, we assume that workers' scale parameters do not deviate at all from the given parameters across the rounds. We call such workers as \emph{stable workers}. In \figref{heter_stable}, we plot the ACT of the compared schemes versus the total UCB by assuming stable workers. In \figref{r_th}, we also present the actual recovery thresholds of SBP coding scheme, MM-GASP, as well as, the average recovery threshold of the rateless coding scheme for different $c$ values. As the name suggests, the rateless scheme does not have a constant recovery threshold. The actual value depends on the computation speeds of the workers, the number of clusters, and the number of workers assigned to them. Therefore, we present the average recovery threshold for this scheme.

\begin{figure}[t]
    \centering
    \resizebox{0.65\linewidth}{!}{
    \input{figures/heter_stable}}
    \caption{ACT vs. total UCB trade-off of the compared schemes with heterogeneous and stable workers.}
    \label{fig:heter_stable}
\end{figure}
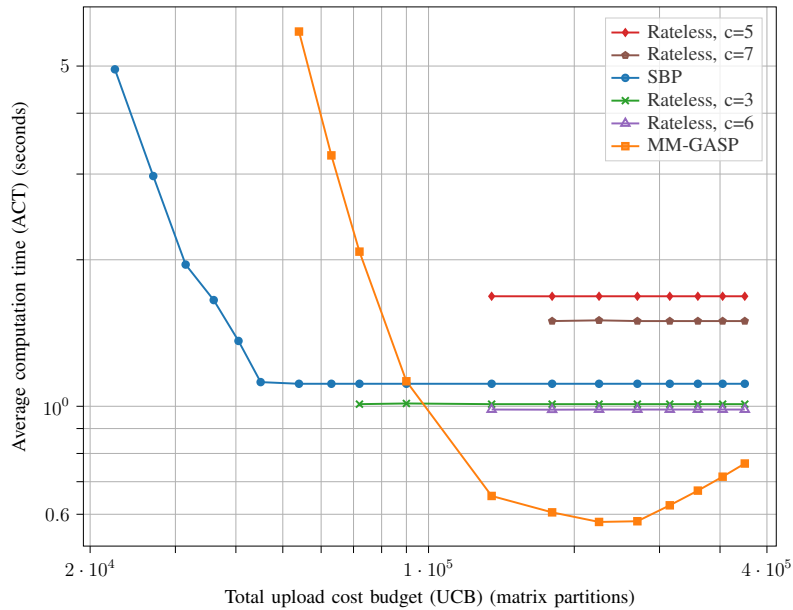

As observed in \figref{heter_stable}, the SBP coding scheme is able to finish the overall task for much lower UCB values than the other two schemes. This is thanks to the fact that SBP is able to provide more computations, $m$, to workers for the same UCB, as highlighted in \remref{model-independ}.
Moreover, although by increasing the total UCB we increase $m$ and therefore $R_{th}$, as observed in \figref{r_th}, and thus workers need to provide more computations to the master, the benefit from having more computations at workers pays off and the ACT decreases when UCB increases. 
The reason for this is the heterogeneity of the workers' speeds.
That is, for a low total UCB, $m$ is so small that the master cannot complete all $R_{th}$ computations from only fast workers. When we increase $m$, the maximum number of computations the fast workers can provide also increases, and the benefit of this increase dominates over the increase in the $R_{th}$. 
For the SBP scheme, this is so until we reach a total UCB value corresponding to $m=L$ i.e., total UCB of $L\times N=45000$. After this point, the ACT of SBP coding scheme stays constant. This is an inherent limitation of SBP coding scheme since the maximum value of $m$ is $L$. Beyond that value, we cannot benefit from the additional UCB.

For  MM-GASP codes, we observe that, although their recovery threshold is close to that of SBP coding scheme in the low UCB regime,
as seen in \figref{r_th}, the minimum total UCB for which MM-GASP codes are able to complete the overall task is larger than SBP coding scheme. That is because the MM-GASP scheme is a univariate scheme; and thus, for the same total UCB, the maximum number of computations a worker can provide is less than the one in SBP coding scheme. For the same reason, at intermediate total UCB availability, i.e., values less than $9\times 10^4$ partitions, the ACT of the MM-GASP scheme is quite large compared to SBP coding scheme. However, for larger values of total UCB, we observe in \figref{heter_stable} that MM-GASP's ACT decreases rapidly, substantially outperforming the other two schemes. However, after a critical point, if the total UCB further increases, then the ACT starts to increase again. After that critical point, the increase in the recovery threshold is not compensated by the additional computations at workers and the ACT starts to increase. 
Unfortunately, operating at this point may not be always possible. 
Especially when we do not have any prior information about the statistics of the workers' speeds. Nevertheless, some heuristics can still be useful to approximate it and even if the optimal point cannot be found, a sufficiently close point can still be beneficial.
Thus, we can conclude that, if a good heuristic can be found to identify a near-optimal $m$ value, for large UCB values, MM-GASP codes can complete the overall task faster than SBP coding scheme as well as the rateless coding scheme. This makes MM-GASP codes a good alternative for the scenarios with high UCB availability.

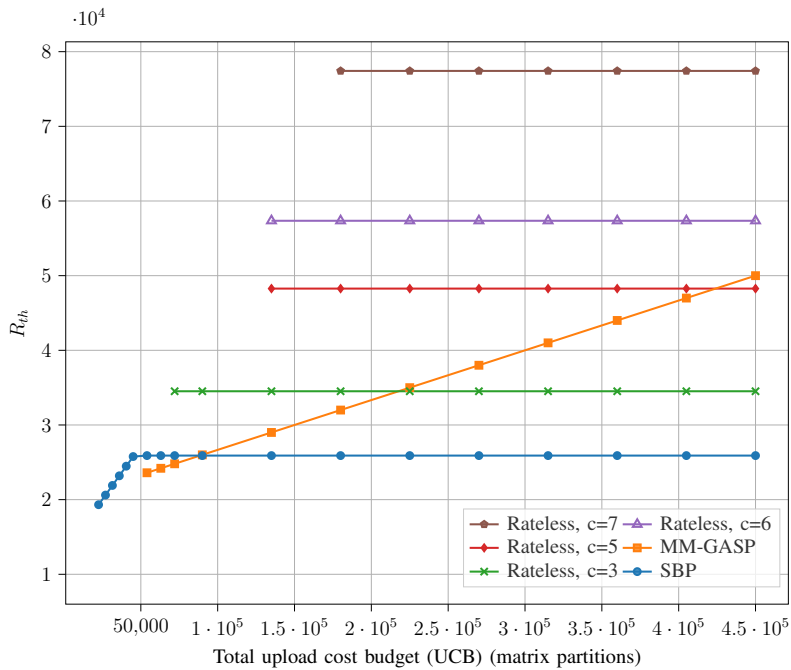
\begin{figure}[t]
    \centering
    \resizebox{0.65\linewidth}{!}{
    \input{figures/r_th_hetero_stable}}
    \caption{Average $R_{th}$ of the compared schemes with heterogeneous and stable workers}
    \label{fig:r_th}
\end{figure}

Finally, for the rateless codes, as observed for MM-GASP codes, we observe that this scheme starts being able to complete the overall task only at a relatively high total UCB value.
That is because the rateless coding scheme assigns a new sub-task to a worker as soon as it finishes its task without waiting for the other clusters to finish. Thus, the UCB is greedily invested in the fastest cluster. However, despite its speed, in terms of the number of useful computations provided, the fastest cluster is less efficient than the other clusters. Please remember that $d_1=\left\lfloor (n_1 -2T+1)/2 \right\rfloor$, but $d_u=\left \lfloor (n_u -T+1)/2 \right \rfloor$ for $2\leq u \leq c$. Therefore, if the number of workers in the fastest cluster is limited, then for the low UCB values, the rateless scheme cannot complete the overall task since it runs out of the necessary upload resources before the master receives the minimum number of useful computations to decode $AB$, which is $KL(1+\epsilon)$. Moreover, we observe in \figref{r_th} that when the number of clusters is low, the recovery threshold is also lower, and the rateless scheme starts completing the overall task at a lower value of total UCB. That is because when $c$ is low, since we assign $N/c$ workers per cluster, there are more workers in the fastest cluster. However, in \figref{heter_stable}, we also observe that this does not always have a positive impact on the ACT.
On the other hand, when UCB is large enough for rateless codes to complete the overall task, its ACT is slightly better than SBP coding scheme for $c=3$ and $c=6$, but for $c=5$ and $c=7$, SBP coding scheme performs better. In general, we expect that the rateless coding scheme performs well when $c$ is equal to the number of heterogeneity classes, but, in this case, we also observe that it performs equally well for $c=3$. That is because, for $c=3$, there is no $\lambda_i$, $i\in[6]$ appearing in more than one cluster, i.e., workers in the same heterogeneity class are allocated to the same cluster. Therefore, for rateless codes, it is important to choose the design parameter $c$ carefully. In practice, we may not know the number of heterogeneity classes, such a clear grouping of computation statistics may not be possible. For such cases, SBP coding scheme or the MM-GASP may be preferable over the rateless coding scheme. 

In addition to choosing $c$ optimally, estimating the instantaneous speeds of the workers is another issue we need to address in rateless codes. In real-world scenarios, the speeds of the workers can occasionally change due to temporary failures, parallel job assignments, etc. To model this, we introduce another simulation scenario, in which workers' scale parameters can deviate from their original values with a very low probability $\rho$. We refer to such workers as \emph{mostly-stable workers}. That is, in any round, a worker with $\lambda_i$ sticks to $\lambda_i$ with probability $1-\rho$, but with a small probability $\rho$, it draws its scale parameter uniform randomly from $\{\lambda_j \mid j\in [6] \}$. We consider such a scenario to model the instantaneous changes in workers' speeds since the detection of such changes by the master and putting this worker to the correct cluster takes at least one round. Taking $\rho=0.001$, we plot the ACT of the compared schemes in \figref{heter_instable}. 

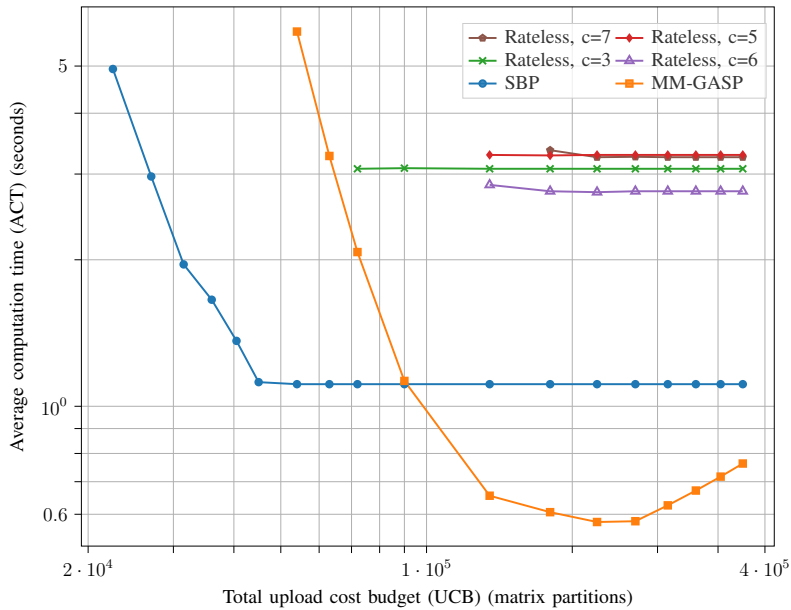
\begin{figure}[t]
    \centering
    \resizebox{0.65\linewidth}{!}{
    \input{figures/heter_instable}}
    \caption{ACT vs. total UCB trade-off of the compared schemes with heterogeneous and mostly-stable workers.}
    \label{fig:heter_instable}
\end{figure}

We observe that even with such a small probability deviation from the estimated scale parameters, the performance of \cite{bitar2020rateless} degrades considerably. Thus, we can argue that, in addition to the substantial improvement in the low and the intermediate UCB values, SBP coding scheme can be advantageous over \cite{bitar2020rateless} in the presence of a high UCB as well depending on the statistics of the workers' speeds.

\subsection{Homogeneous Workers}

In this subsection,  we assume that the computation speeds of the workers are homogeneous, and we compare the ACTs of the considered schemes with respect to the available total UCB. That is, we have 450 workers as in \subsecref{heter}, but this time, all the workers follow the same computation statistics with $\lambda=10^{-2}\times KL$ and $\nu= 10/(KL)$. We assume at most $T=N/15$ workers can collude, and we divide $A$ and $B$ into $K=L=100$ partitions. 
For the rateless scheme in  \cite{bitar2020rateless}, although, the workers' speeds are homogeneous, we consider different numbers of clusters  $c\in [3]$ in order to analyse its effect. In \figref{homo_stable}, we present the ACT versus UCB plot for this setting. 

\begin{figure}[t]
    \centering
    \resizebox{0.65\linewidth}{!}{
    \input{figures/homo_stable}}
    \caption{ACT vs. total UCB trade-off of the compared schemes with homogeneous and stable workers.}
    \label{fig:homo_stable}
\end{figure}
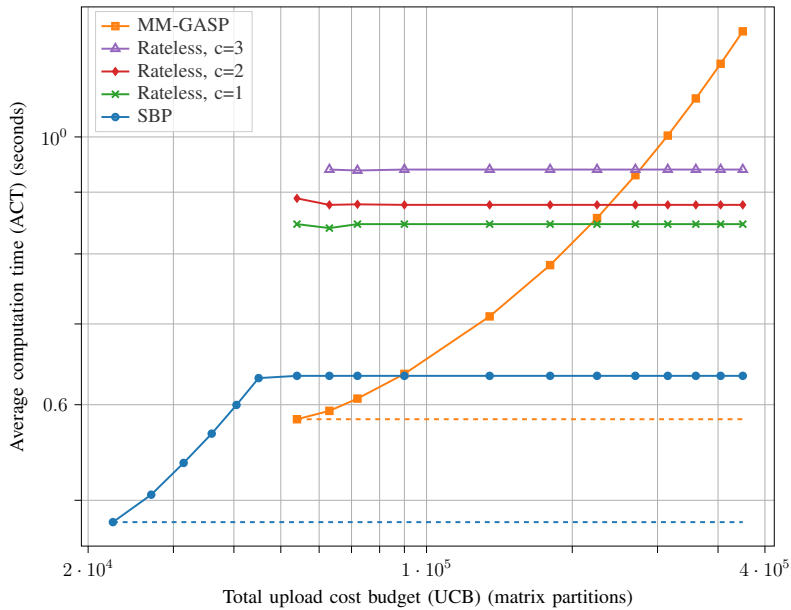

Similarly to the heterogeneous case discussed in \subsecref{heter}, due to the upload cost efficiency of the bivariate polynomial codes, we observe that the minimum UCB for which SBP can complete the overall task is smaller than for the other schemes. Moreover, in this homogeneous case, we observe that the ACTs of SBP and MM-GASP only increase with the total UCB. That is because, due to the similarity in workers' speeds, there is no need for the faster workers to compensate for the slower ones. Therefore, rather than improving the ACT, increasing $m$ beyond the minimum value, for which the schemes complete the overall task, results in a  higher ACT since it also increases $R_{th}$. Therefore, we depict the best ACT for SBP and MM-GASP coding schemes in \figref{homo_stable} and \figref{homo_mostly_stable} by flat dashed lines. 

On the other hand, for the rateless codes, we observe that, regardless of the number of clusters, $c$, considered, they perform significantly worse than SBP coding scheme for all UCB values. That is because, while the sub-tasks are one-to-any replaceable in SBP coding scheme, i.e, the result of any sub-task can compensate for the absence of any other sub-task, this is not the case in the rateless coding scheme. Since we consider the homogeneous workers in their speeds, there is not much difference between the clusters in the rateless coding scheme. Since, to decode the sub-tasks in a cluster, all of the workers in that cluster must finish their sub-tasks, the ACT increases. 

As we stated in \subsecref{heter}, in a real-world scenario, the speeds of workers can occasionally change. To model this effect, in \figref{homo_mostly_stable}, we provide the ACT versus UCB trade-off in the scenario in which the workers are mostly-stable with a transition probability $\rho=0.001$. Since there is only one heterogeneity class in the homogeneous case, to simulate mostly-stable workers, we assume that a worker sticks to $\lambda=10^{-2}\times KL$ with probability $\rho$, but with probability $1-\rho$, its $\lambda$ parameter is chosen uniformly between $\lambda=10^{-3}\times KL$ and $\lambda=10^{-4} \times KL$.

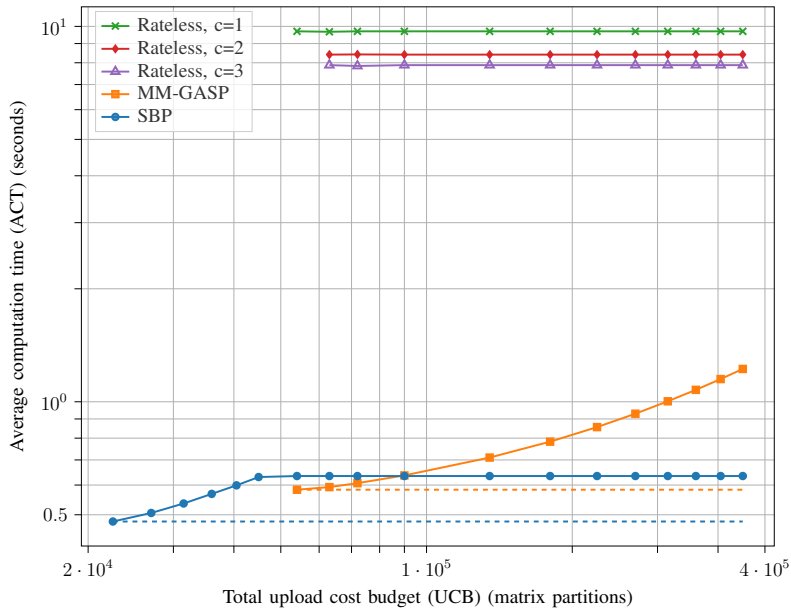
\begin{figure}[t]
    \centering
    \resizebox{0.65\linewidth}{!}{
    \input{figures/homo_mostly_stable}}
    \caption{ACT vs. total UCB trade-off with homogeneous and mostly-stable workers.}
    \label{fig:homo_mostly_stable}
\end{figure}

We observe that the effect of such a low probable deviation from the original parameters is considerable in the rateless codes since in order to utilize the computations in a cluster, all the workers in that cluster must complete their sub-tasks. If some of these workers straggle even only for one round, it can delay the overall computation significantly.

To conclude, we observe that, in the cases in which the workers' speeds are known to be close to each other, i.e., homogeneous, SBP coding scheme is preferable over both the rateless coding and the MM-GASP schemes.

\section{Proof of \thmref{r_th}} \label{sec:main_proof}

In \figref{poly-coeffs}, we visualize the degrees of the monomials
of $A(x)B(x,y)$ in the $\deg(x)-\deg(y)$ plane. From \figref{poly-coeffs},
we see that the number of monomials of $A(x)B(x,y)$ is $(K+T)L+m(K+T-1)$.
We need to show that every possible combination of so many responses
from the workers interpolates to a unique polynomial, implying $(K+T)L+m(K+T-1)$
is the recovery threshold.

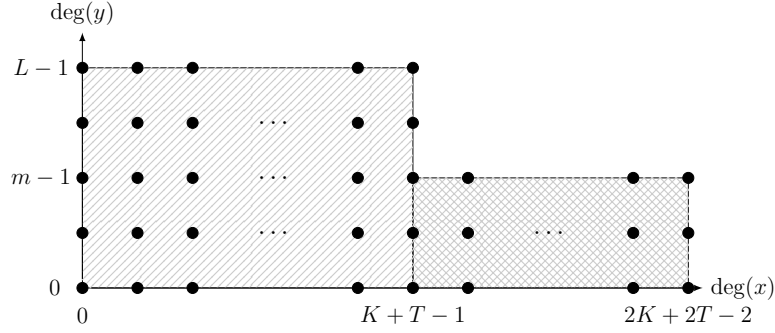
\begin{figure}[t]
\centering 
\resizebox{0.65\linewidth}{!}{
\input{figures/degxdegy}}
\caption{\label{fig:poly-coeffs}The visualization of the degrees of the monomials
of $A(x)B(x,y)$ in the $\deg(x)-\deg(y)$ plane.}
\end{figure}

\begin{defn}
Bivariate polynomial interpolation problem can be formulated as solving
a linear system of equations, whose unknowns are the coefficients
of $A(x)B(x,y)$ and whose coefficient matrix consists of the monomials
of $A(x)B(x,y)$ and their derivatives with respect to $y$ evaluated at the evaluation points of the workers, \revision{$(x_i,y_i),i\in[N]$}. We refer
to this coefficient matrix as the \textbf{interpolation matrix} and denote
it by $M$. Since the number of monomials of $A(x)B(x,y)$ is $R_{th}=(K+T)L+m(K+T-1)$,
we require \revision{$R_{th}$ equations to interpolate it, and hence,} $M\in\mathbb{R}^{R_{th}\times R_{th}}$. Each row of $M$
corresponds to the result of one sub-task sent by a worker to the
master. For example, when $K=L=2$, $m=2$ and $T=1$, we have  $R_{th}=10$, and one possible interpolation matrix formed by any 5 workers, each of which provides $m=2$ computations, is as follows:
$$\displaystyle
M=\begin{bmatrix}1 & x_{1} & x_{1}^{2} & x_{1}^{3} & x_{1}^{4} & y_{1} & x_{1}y_{1} & x_{1}^{2}y_{1} & x_{1}^{3}y_{1} & x_{1}^{4}y_{1}\\
0 & 0 & 0 & 0 & 0 & 1 & x_{1} & x_{1}^{2} & x_{1}^{3} & x_{1}^{4}\\
\vdots & \vdots & \vdots & \vdots & \vdots & \vdots & \vdots & \vdots & \vdots & \vdots\\
1 & x_{5} & x_{5}^{2} & x_{5}^{3} & x_{5}^{4} & y_{5} & x_{5}y_{5} & x_{5}^{2}y_{5} & x_{5}^{3}y_{5} & x_{5}^{4}y_{5}\\
0 & 0 & 0 & 0 & 0 & 1 & x_{5} & x_{5}^{2} & x_{5}^{3} & x_{5}^{4}
\end{bmatrix}.$$

Observe that the first row represents the direct evaluation $A(x_{1})B(x_{1},y_{1})$
from worker 1, and the second row represents $A(x_{1})\partial_{1}B(x_{1},y_{1})$,
again from worker 1. In general, any interpolation matrix formed by $R_{th}=10$ computations received from any subset of workers is also valid, as long as the workers follow the computation order specified in Section \ref{subsec:Computation}.
\end{defn}

The problem of showing that any $R_{th}$ responses from the workers
interpolates to a unique polynomial is equivalent to showing that
the corresponding interpolation matrix is non-singular. The theorem
claims that this is the case with high probability. First, we need
to show that there exist some evaluation points for which the determinant
of the interpolation matrix is not zero. That is equivalent to showing
that $\det(M)$ is not the zero polynomial of the evaluation points. In \cite{hasircioglu2020bivariate}, such a result for the same type
of interpolation matrices is shown for the real field $\mathbb{R}$.
Here, we extend this proof to $\mathbb{F}$. We show that $\det(M)$
is non-zero for some evaluation points by using Taylor series expansion
of $\det(M)$, as done in \cite{hasircioglu2020bivariate}. This can
be done since Taylor series expansion is also applicable in $\mathbb{F}$, as long as, the degree of the polynomial $A(x)B(x,y)$ is smaller than the
field size $q$. This can be guaranteed by choosing a large $q$.
For further details on the applicability of Taylor series expansion in
finite fields, see \cite{hoffmanlinear} and \cite{felix_fontein_2009}. 

Without losing generality, let us assume first that $n$ workers with $n\leq N$, provide, together, enough responses, i.e., $R_{th}$, to interpolate $A(x)B(x,y)$. Let us assume $(x_{i},y_{i})$ and $(x_{j},y_{j})$ are two evaluation
points for which the evaluations of $A(x)B(x,y)$ and some of its
derivatives at these points are received by the master. We write the
Taylor series expansion of $\det(M)$ around $(x_{i},y_{i})$ by taking
the evaluation point $(x_{j},y_{j})$ as the variable: 
\begin{equation}
\det(M)=\sum_{(\alpha_{1},\alpha_{2})\in\mathbb{N}^{2}}\frac{1}{\alpha_{1}!\alpha_{2}!}(x_{j}-x_{i})^{\alpha_{1}}(y_{j}-y_{i})^{\alpha_{2}}D_{\alpha_{1},\alpha_{2}}(\tilde{Z}),\label{eq:taylor_expansion}
\end{equation}
where $\tilde{Z}\triangleq \{(x_k,y_k):k\in[n]\}\setminus \{(x_j,y_j)\}$ and $
D_{\alpha_{1},\alpha_{2}}(\tilde{Z})\triangleq\frac{\partial^{\alpha_{1}+\alpha_{2}}}{\partial x_{j}^{\alpha_{1}}\partial y_{j}^{\alpha_{2}}}\det(M)(x_{j},y_{j})\biggl|_{x_{j}=x_{i},y_{j}=y_{i}}.$
We call $(x_{i},y_{i})$ the \textbf{pivot node} and $(x_{j},y_{j})$
the \textbf{variable node}. 
\begin{rem}
The monomials $(x_{j}-x_{i})^{\alpha_{1}}(y_{j}-y_{i})^{\alpha_{2}}$
are linearly independent for different $(\alpha_{1},\alpha_{2})$
pairs if there is no relation between $x$ and $y$ coordinates of
the evaluation points, i.e., $x_{i}$ and $x_{j}$ do not depend on
$y_{i}$ and $y_{j}$. Thus, $\det(M)$ is a zero polynomial of all
evaluation points, if and only if $D_{\alpha_{1},\alpha_{2}}(\tilde{Z})=0,\forall(\alpha_{1},\alpha_{2})\in\text{\ensuremath{\mathbb{N}^{2}}}$. Therefore,
in order to show that $M$ is non-singular, it suffices to show that there exists 
at least one $(\alpha_{1},\alpha_{2})$ making $D_{\alpha_{1},\alpha_{2}}(\tilde{Z})$ nonzero. 
\end{rem}
Before looking into $D_{\alpha_{1},\alpha_{2}}(\tilde{Z})$
in more detail, let us define some notions which will help us understanding
its structure.

\begin{defn}
\label{def:InterpolationSpace} \textbf{Derivative Set.} In an interpolation matrix $M$, there might be several rows each corresponding to a different derivative order  of $A(x)B(x,y)$ associated with the evaluation point $z_{i}\triangleq(x_{i},y_{i})$,
which is assigned to worker $i$. We define the derivative set of
$z_{i}$, denoted by $U_{z_{i},M}$ as the set of derivative orders of
$A(x)B(x,y)$ with respect to $x$ and $y$ associated to $z_i$ in $M$. That is, $(d_x,d_y)\in U_{z_{i},M}$ if and only if $M$ has a row corresponding to $\frac{\partial^{d_x+d_y}}{\partial x_i^{d_x} \partial y_i^{d_y}}A(x_i)B(x_i,y_i)$.
\end{defn}

\begin{defn}\textbf{Derivative order space.} The derivative
order space of a bivariate polynomial $A(x)B(x,y)$ is defined as
the 2-dimensional space of all its possible derivative orders. Since
the largest derivative order of a bivariate polynomial is its largest
monomial degree, the derivative order space has the same shape as
$\deg(x)-\deg(y)$ plane depicted in \figref{poly-coeffs}. For example,
consider the 2-D derivative order $(K+T,m)$. Since all the monomials
of $A(x)B(x,y)$ having a degree larger than $m$ with respect to
$y$ have a degree less than or equal to $K+T-1$ with respect to $x$,
the 2-D derivative order $(K+T,m)$ results in a zero polynomial. Thus,
this is not an element of derivative order space. The derivative set
of each evaluation point can be depicted in the derivative order space
separately.
\end{defn}
\begin{defn}
\textbf{\label{def:operator}}Let $M$ be an interpolation matrix
for which some of its rows depend on $x_{j}$ and $y_{j}$. Let
us denote by $r_{i}$ its $i^{th}$ row and define \revision{a \emph{simple shift} as
\begin{equation}
\partial_{i,x_j}M\triangleq\left[r_{1}^{T},\dots,\frac{\partial}{\partial x_{j}}r_{i}^{T},\dots,r_{KL}^{T}\right]^{T}
\end{equation}
and 
\begin{equation}
\partial_{i,y_j}M\triangleq\left[r_{1}^{T},\dots,\frac{\partial}{\partial y_{j}}r_{i}^{T},\dots,r_{KL}^{T}\right]^{T}.
\end{equation}
That is, $\partial_{i,x_j}$ and $\partial_{i,y_j}$ transform $M$ into another matrices by
taking the derivative of its $i^{th}$ row with respect to $x_{j}$ and $y_j$, respectively. If in the resulting derivative set, i.e., $U_{x_j,\partial_{i,x_j}M}$ or $U_{y_j,\partial_{i,y_j}M}$, each element has a multiplicity of one, then the shift is called a \emph{regular simple shift}.
\end{defn}}
\revision{
\begin{defn}
Let $\boldsymbol{k}$ and $\boldsymbol{l}$ be vectors such that $\boldsymbol{k}\in\{0,1,\cdots,K-1\}^{R_{th}}$ and $\boldsymbol{l}\in\{0,1,\cdots,L-1\}^{R_{th}}$. We define \emph{the composition of simple shifts} as
\begin{equation}
\label{eq:composition_of_simple_shifts}
\nabla^{x_j,y_j}_{\boldsymbol{k},\boldsymbol{l}}M=\partial_{1,x_j}^{\boldsymbol{k}(1)}\partial_{2,x_j}^{\boldsymbol{k}(2)}\cdots\partial_{R_{th},x_j}^{\boldsymbol{k}(R_{th})} \partial_{1,y_j}^{\boldsymbol{l}(1)} \partial_{2,y_j}^{\boldsymbol{l}(2)} \cdots \partial_{R_{th},y_j}^{\boldsymbol{l}(R_{th})} M.
\end{equation}
\end{defn} 

That is, the $i^{th}$ element of $\boldsymbol{k}$ denotes the order of the derivative of $i^{th}$ row of $M$ with respect to the variable $x_j$, and the $i^{th}$ element of $\boldsymbol{l}$ denotes the order of the derivative of $i^{th}$ row of $M$ with respect to the variable $y_j$. In fact, \eqref{composition_of_simple_shifts} is not the only way to compute $\nabla^{x_j,y_j}_{\boldsymbol{k},\boldsymbol{l}}M$ since the derivative operation is commutative. One can compute  $\nabla^{x_j,y_j}_{\boldsymbol{k},\boldsymbol{l}}M$ in any other order. Each of these possible orders are referred to as a derivative path. If a derivative path involves only regular simple shifts, i.e. after each derivative there are not two equal rows, then it is called a \emph{regular derivative path}. We denote the number of regular derivative paths by $C_{\boldsymbol{k},\boldsymbol{l}}(M)$.
}

%

Based on these definitions, we have the following lemma. 

\revision{
\begin{lem}[Lemma 1 in \cite{hasircioglu2020bivariate}]
\label{lem:regular_permutations}
Let $\boldsymbol{k}\in\{0,1,\cdots,K-1\}^{R_{th}}$, $\boldsymbol{l}\in\{0,1,\cdots,L-1\}^{R_{th}}$ and $\alpha_{1}=\sum_{i=1}^{R_{th}}\boldsymbol{k}(i)$
and $\alpha_{2}=\sum_{i=1}^{R_{th}}\boldsymbol{l}(i)$. Then, we have
\begin{equation}
\label{eq:derivative_lemma}
\frac{\partial^{\alpha_{1}+\alpha_{2}}}{\partial x_j^{\alpha_{1}}\partial y_j^{\alpha_{2}}} \det(M)\biggl|_{x_{j}=x_{i},y_{j}=y_{i}}=\sum_{
\left(\boldsymbol{k,l}\right)\in \mathcal{R}_M(\alpha_1,\alpha_2)}C_{\boldsymbol{k},\boldsymbol{l}}(M)\det\left(\nabla^{x_j,y_j}_{\boldsymbol{k},\boldsymbol{l}}M\right)\biggl|_{x_{j}=x_{i},y_{j}=y_{i}},
\end{equation}
where $\mathcal{R}_M(\alpha_1,\alpha_2)$ is the set of  $(\boldsymbol{k},\boldsymbol{l})$ pairs satisfying  $C_{\boldsymbol{k},\boldsymbol{l}}(M)\neq 0$, i.e., there is at least one derivative path for which $\nabla^{x_j,y_j}_{\boldsymbol{k},\boldsymbol{l}}$ can be applied by using only regular simple shifts. 
\end{lem}

\begin{defn}
If $\mathcal{R}_M(\alpha_1,\alpha_2)$ has only one element, i.e., there is only one $(\boldsymbol{k},\boldsymbol{l})$ resulting in a regular simple shift, then $(\alpha_1,\alpha_2)$ is called a \emph{unique shift order} and $(\boldsymbol{k},\boldsymbol{l})$ is called a \emph{unique shift}.
\end{defn}
}

Now, let us go back to $D_{\alpha_{1},\alpha_{2}}(\tilde{Z})$. Recall
that we would like to show that at least for one $(\alpha_{1},\alpha_{2})$, $D_{\alpha_{1},\alpha_{2}}(\tilde{Z})$ is non-zero. Observe that
it does not depend on $(x_{j},y_{j})$ since after taking the derivatives
with respect to $(x_{j},y_{j})$, the resulting expression is evaluated at $x_{j}=x_{i},y_{j}=y_{i}$.
If $(\alpha_{1},\alpha_{2})$ is a unique shift order, then according
to \eqref{derivative_lemma}, it is enough to show that $\det\left(M_{1}\right)$ is not
the zero polynomial, where $M_{1}\triangleq \nabla^{x_j,y_j}_{\boldsymbol{k},\boldsymbol{l}} M\mid_{x_{j}=x_{i},y_{j}=y_{i}}$.  Notice that,
$M_{1}$ no longer depends on the evaluation point $(x_{j},y_{j})$. \revision{We call such a procedure of transforming an interpolation matrix into another interpolation matrix via unique shifts as the \emph{coalescence} of the variable node and the pivot node. After obtaining $M_1$, we can employ the same idea to show $M_1$ is non-singular. Namely, we can} write the Taylor series expansion of $\det(M_{1})$ by choosing
a new variable node and keeping the same pivot node $(x_{i},y_{i})$. \revision{If there is a unique shift for the coalescence, the resultant matrix $M_2$ will not depend on neither the previous variable node $(x_j,y_j)$ nor the current variable node.} We can apply such coalescences successively as long as we can find a unique
shift order $(\alpha_{1},\alpha_{2})$ at each coalescence,
until $M_{\text{final}}$ depends only on
one evaluation point, which is the pivot node, $(x_i,y_i)$.
In $M_{\text{final}}$, the derivative
set of $(x_{i},y_{i})$ has all possible elements of the derivative
order space. Thus, $M_{\text{final}}$ is a triangular matrix, and hence, non-singular. 

\revision{To summarize, to prove that all possible interpolation matrices, $M$, generated from our scheme are non-singular in general,
we need to show that we can always find at least one unique shift for all the steps of the coalescence
procedure.} Our strategy to show that we can always find a unique shift in all coalescence steps is based
on the idea of keeping the derivative set of the pivot node to be a lower
set. \revision{A lower set is defined as a set in which the presence of an element implies the presence of all possible elements smaller than this element. To decide if an element is smaller than any other element, we need to define an ordering rule. For our case,} we define such an ordering as follows. Assume we
denote our pivot node as $z_{i}=(x_{i},y_{i})$ and take two derivative
orders $(a,b)\in U_{z_{i},M}$ and $(c,d)\in U_{z_{i},M}$, where $a$
and $c$ are the orders of the derivative with respect to $x_{i}$
and $b$ and $d$ are the orders of the derivative with respect to
$y_{i}$. We say $(a,b)<(c,d)$ if and only if $a<c$ or $a=c$ and
$b<d$.

Before formally stating our strategy to find a unique shift
in all the coalescence steps, we describe it with two examples.

\begin{example}
\label{exa:coal_ex1}Assume $K=L=5$, $T=1$ and $m=3$ \revision{and we are at the beginning of $p$-th coalescence step.} Let us choose $z_{i}$
as the pivot node and $z_{j}$ as the variable node. Further, assume \revision{at the beginning we have $U_{z_{i},M_{p-1}}=\{(a,b):(a,b)\leq(1,2)\}$ }
and $U_{z_{j},M_{p-1}}=\{(a,b):(a,b)\leq(0,2)\}$. We depict the derivative
sets of $U_{z_{i},M_{p-1}}$ and $U_{z_{j},M_{p-1}}$ in \figref{coal_ex1_piv} and
\figref{coal_ex1_var}. We will take smallest possible shift $(\alpha_{1},\alpha_{2})$
such that the resultant $U_{z_{i},M_p}$ after the coalescence. 
Knowing the number of elements
in $U_{z_{i},M_p}$ after the coalescence, its shape is uniquely determined under the condition that it must be a lower set and shown in \figref{coal_ex1}. In \figref{coal_ex1_var} and \figref{coal_ex1}, we assign to each element of \revision{ $U_{z_{j},M_{p-1}}$ either the letter "a", "b" or "c" so that we can track its location during and after the coalescence procedure.}
\revision{Recall that taking derivatives corresponds to shifting the elements of the derivative set in the derivative
order space. Thus, in order to shift the elements of $U_{z_{j},M_{p-1}}$ to their locations in the final shape in \figref{coal_ex1},
we need to the total number of shifts in both $x$ and $y$ directions is $4$, implying we need to choose $(\alpha_{1},\alpha_{2})=(4,4)$. For this choice, we have $\boldsymbol{k}(i_a)=2$,  $\boldsymbol{k}(i_b)=1$, $\boldsymbol{k}(i_c)=1$, $\boldsymbol{l}(i_a)=0$, $\boldsymbol{l}(i_b)=2$ and $\boldsymbol{l}(i_c)=2$ where $i_a$ is the row-index of the element $a$.}
Given this choice of $(\alpha_{1},\alpha_{2})$, there is no other
possible resulting shape for $U_{z_{i},M_p}$ resulting a non-singular $M_p$. To see this, observe that, if we write the derivative sets of $U_{z_{i},M_p}$ after-the-coalescence for all possible $(\boldsymbol{k,l})$
such that $\sum_{i}\boldsymbol{k}(i)=\alpha_{1}$ and $\sum_{i}\boldsymbol{l}(i)=\alpha_{2}$, then all, except the one depicted in \figref{coal_ex1} will have overlapping elements making the corresponding interpolation matrix singular.
Therefore, $(\alpha_{1},\alpha_{2})=(4,4)$ is a unique
shift order.

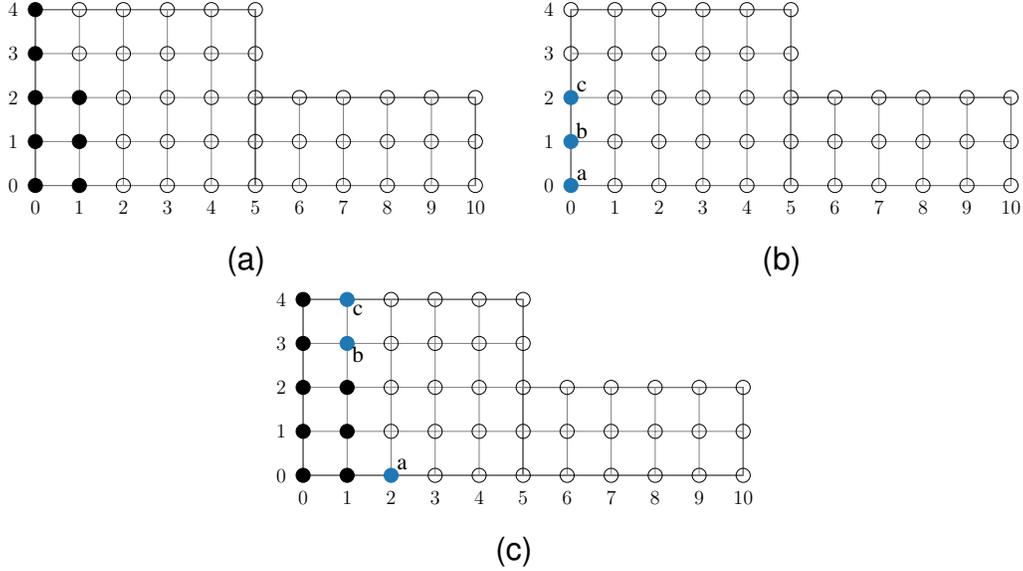
\begin{figure*}
\centering
\subfloat[\label{fig:coal_ex1_piv}]{
\scalebox{0.9}{\input{figures/coalescence_ex1_piv}}}\hspace{5pt}
\subfloat[\label{fig:coal_ex1_var}]{
\scalebox{0.9}{\input{figures/coalescence_ex1_var}}}\hspace{5pt}
\subfloat[\label{fig:coal_ex1}]{
\scalebox{0.9}{\input{figures/coalescence_ex1}}}
\caption{Depictions of the derivative sets in \exaref{coal_ex1}.}
\end{figure*}









\end{example}
\vspace{-5pt}
\begin{example}
\label{exa:coal_ex2}Let us consider the same setting in \exaref{coal_ex1}, but assume $U_{z_{i},M_{p-1}}=\{(a,b):(a,b)\leq(6,1)\}$. Since the maximum number
of computations a worker can provide is $m=3$, the cardinality of the
derivative set of the variable node $U_{z_j,M_{p-1}}$, in this example, is at its maximum.
Thus, we can directly follow the same procedure as in \exaref{coal_ex1}.
Note that after the coalescence, $U_{z_{i}M_p}$ will have 34 elements,
and the lower set having 34 elements is unique and well defined. To
obtain the shape in \figref{coal_exa2}, we need $(\alpha_{1},\alpha_{2})=(0,19)$ \revision{with $\boldsymbol{k}(i_a)=7$,  $\boldsymbol{k}(i_b)=6$ and $\boldsymbol{k}(i_c)=6$,
and it is a unique shift order since any other assignment of 19 shifts to $a$, $b$ and $c$ results in a non-singular $M_p$.}

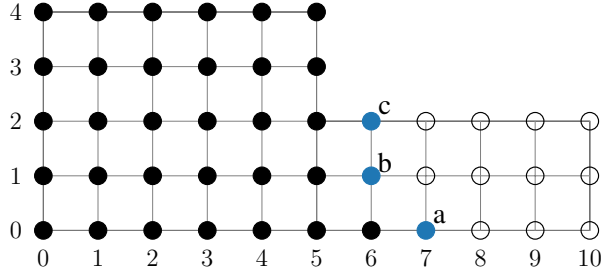
\begin{figure}
\centering
\resizebox{0.5\linewidth}{!}{\input{figures/coalescence_ex2}}

\caption{\label{fig:coal_exa2}Depiction of $U_{z_{i},M_p}$ after coalescence
in \exaref{coal_ex2}}
\end{figure}
\end{example}

\begin{figure*}
\centering
\subfloat[\label{fig:proof_fig1}]{
\scalebox{0.9}{\input{figures/proof_fig1}}}
\subfloat[\label{fig:proof_fig2}]{
\scalebox{0.9}{\input{figures/proof_fig2}}}\\
\subfloat[\label{fig:proof_fig3}]{
\scalebox{0.9}{\input{figures/proof_fig3}}}
\caption{Visualization of the pivot node and the variable node during a coalescence}
\end{figure*}
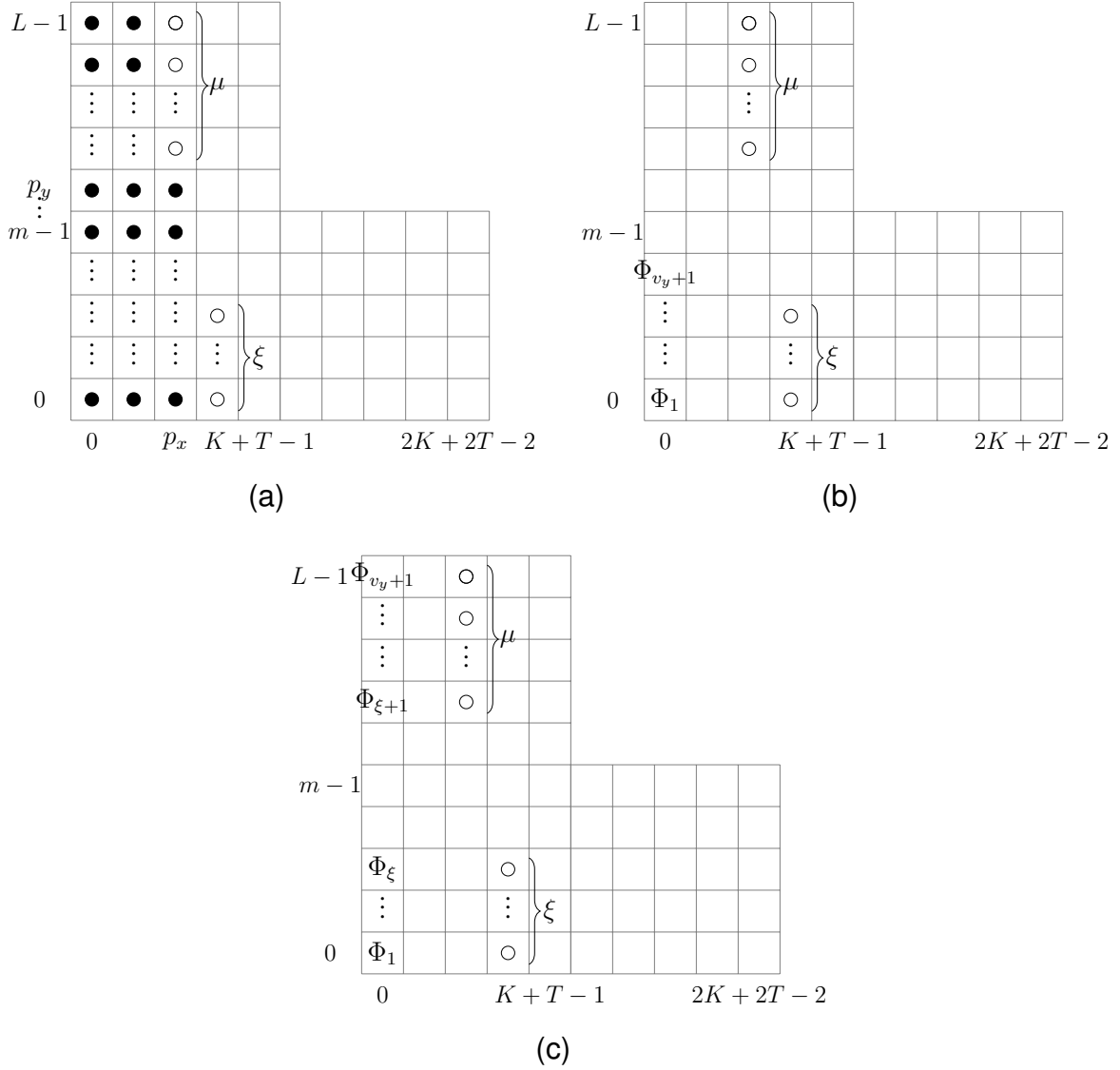

Next, we formally state our strategy \revision{for an arbitrary coalescence step $p$.} Since we choose one pivot node and use it for every coalescence step, we guarantee
that the variable node's derivative set has always at most $m$ elements.
To generalize the procedure in \exaref{coal_ex1} and \exaref{coal_ex2}, let us assume $(p_{x},p_{y})$ is the largest element of the derivative
set of the pivot node $z_{i}$,
 i.e., $U_{z_{i},M_{p-1}}=\{(a,b):(0,0)\leq(a,b)\leq(p_{x},p_{y})\}$ and
$(0,v_{y})$ is the largest element of the derivative set of the variable node
$z_{j}$, i.e., $U_{z_{j},M_{p-1}}=\{(0,b):0\leq b\leq v_{y}\}$. While calculating
$(\alpha_{1},\alpha_{2})$ pair, we first determine $\alpha_{2}$,
which is the total derivative order with respect to $y_{j}$, or equivalently
the number of shifts towards $y$-direction in the derivative order
set. \revision{This means that we first take the derivatives with respect to $y_i$, and then with respect to $x_i$. In \figref{proof_fig1}, in the derivative order space, for $p_x\leq K+T-1$, we depict the derivative set of the pivot node, i.e., $U_{z_i,M_{p-1}}$, by filled circles, and the locations to which the elements of $U_{z_j,M_{p-1}}$ will be placed after the coalescence by unfilled circles. Note that we determine these locations by inserting the elements of $U_{z_j,M_{p-1}}$ into $U_{z_i,M_{p-1}}$ such that the derivative set of the pivot after the coalescence, i.e., $U_{z_i,M_p}$, satisfies the lower set property. In \figref{proof_fig2}, instead of the elements of $U_{z_i,M_{p-1}}$, we depict the elements of $U_{z_j,M_{p-1}}$ together with the locations they will be placed after the coalescence to facilitate visualization of the necessary shifts. To be able to keep track of the elements, we depict each one of them by $\Phi_i$, $i\in[v_y+1]$.} We denote the number of elements in $U_{z_{j},M_{p-1}}$ to be shifted towards
$y$-direction, by $\mu$. \revision{We further define $\xi\triangleq v_y+1-\mu$. As shown in \figref{proof_fig1}, when the number of
empty spaces in the rightmost partially occupied column of $U_{z_{i},M_{p-1}}$ is smaller than $|U_{z_j,M_{p-1}}|=v_y+1$, $\mu$ becomes this number, i.e., $\mu=L-p_y-1$, since these spaces must be filled. Otherwise, to fill as many as spaced possible, all elements of the derivative set of the pivot node are shifted towards $y$-direction and $\mu$ becomes $v_y+1$.}
Thus, $\mu=\min\{L-p_{y}-1,v_{y}+1\}$ if $p_{x}\leq K+T-1$. \revision{When $p_x>K+T-1$, in fact, the same logic also applies but the maximum number of elements that can be placed in a column in \figref{proof_fig1} would be $m$ instead of $L$. Thus, the expression for $\mu$ is modified as $\mu=\min\{m-p_{y}-1,v_{y}+1\}$, which is obtained by replacing $L$ with $m$.}

\revision{Next, recall that only regular simple shifts are considered for unique shifts. Thus, while taking $y$-directional derivatives, i.e., shifts towards $y$-direction in \figref{proof_fig2}, the sequence of the elements in the $y$-axis does not change. For instance, $\Phi_{v_y+1}$ stays always on top of the elements denoted by $\Phi_i,i\in[v_y]$. If, for example, as a result of some shifts, $\Phi_{v_y}$ is placed on top of $\Phi_{v_y+1}$, then this would be possible only if the element $\Phi_{v_y}$ is located in the same location as $\Phi_{v_y+1}$ at some point, and this would contradict the assumption of regular simple shifts. Hence, there is only one resulting order after shifting the uppermost $\mu$ elements towards $y$-direction. We show the elements of the variable node's derivative set after $y$-directional shifts in \figref{proof_fig3}. All the remaining shifts, now, are the ones towards $x$-direction so that the elements of $U_{z_j,M_{p-1}}$ are located to their intended locations, i.e., unfilled circles in \figref{proof_fig3}. Notice that each $\Phi_i$ is already aligned with its final location in $y$-direction, and hence, each one of them will be shifted towards $x$-direction by a sufficient amount. Therefore, these shifts also result in a unique shape. From these observations, we can conclude that whenever the derivative sets of the pivot node and the variable node are lower sets, there exists a unique shift for their coalescence.}

From this discussion, we can conclude that $\det(M)$ is not zero
polynomial for large enough $q$. Next, we need to find the upper
bound on the probability $\det(M)=0$, when the evaluation points
are sampled uniform randomly from $\mathbb{F}$.
\begin{lem}
\label{lem:schwartz_lemma}\textbf{\emph{\cite[Lemma 1]{schwartz1980fast}}}
Assume $P$ is a non-zero, $v$-variate polynomial of variables $\alpha_{i},i\in[v]$.
Let $d_{1}$ be the degree of $\alpha_{1}$ in $P(\alpha_{1},\dots,\alpha_{v})$,
and $P_{2}(\alpha_{2},\dots,\alpha_{v})$ be the coefficient of $\alpha_{1}^{d_{1}}$in\textup{
$P(\alpha_{1},\dots,\alpha_{v})$. }\textup{\emph{Inductively, let
$d_{j}$ be the degree of $\alpha_{j}$ in $P_{j}(\alpha_{j},\dots,\alpha_{v})$
and $P_{j+1}(\alpha_{j+1},\dots,\alpha_{v})$ be the coefficient of
$\alpha_{j}$ in $P_{j}(\alpha_{j},\dots,\alpha_{v})$. Let $S_{j}$
be a set of elements from a field $\mathbb{F}$, from which the coefficients
of $P$ are chosen. Then, in the Cartesian product set $S_{1}\times S_{2}\times\dots\times S_{v}$,
$P(\alpha_{1},\dots,\alpha_{v})$ has at most $\left|S_{1}\times S_{2}\times\dots\times S_{v}\right|\left(\frac{d_{1}}{|S_{1}|}+\frac{d_{2}}{|S_{2}|}+\dots+\frac{d_{v}}{|S_{v}|}\right)$
zeros.}} 
\end{lem}
In our case, $\det(M)$ is a multivariate polynomial of the evaluation points $(x_{i},y_{i})$ since the elements of $M$ are the monomials of $A(x)B(x,y)$
and their derivatives with respect to $y$, evaluated at some $(x_{i},y_{i})$.
Thus, $v$ is the number of different evaluation points in $M$. We
choose the evaluation points from the whole field $\mathbb{F}$.
Thus, $S_{j}=\mathbb{F}$ and $|S_{j}|=q,\forall j\in[v]$, and
$\left|S_{1}\times S_{2}\times\dots\times S_{v}\right|=q^{v}$. Then,
the number of zeros of $\det(M)$ is at most $q^{v-1}(d_{1}+d_{2}+\dots+d_{v})$.
If we sample the evaluation points uniform randomly, then the probability
that $\det(M)=0$ is $(d_{1}+d_{2}+\dots+d_{v})/q$, since we sample
a $v$-tuple of evaluation points from $S_{1}\times S_{2}\times\dots\times S_{v}$.
To find $d_{1}+d_{2}+\dots+d_{v}$, we resort to the definition of
determinant, that is $\det(M)=\sum_{i=1}^{R_{th}}(-1)^{1+i}m_{1,i}M_{1,i}$,
where $m_{1,i}$ is the element of $M$ at row 1 and column $i$ and
$M_{1,i}$ is the minor of $M$ when row 1 and column $i$ are removed
\cite[Corollary 7.22]{liesen_linear_2015}. Thus, to identify the
coefficients in \lemref{schwartz_lemma}, in the first row of $M$,
we start with the monomial with the largest degree. Assuming the monomials
are placed in an increasing order of their degrees, the largest degree
monomial is at column $R_{th}$. If that monomial is univariate, then
$d_{1}$ is the degree of the monomial and the coefficient of $\alpha_{1}^{d_{1}}$
is $P_{2}(x_{2},\dots,x_{v})=\det(M_{1,1})$. If the monomial is bivariate,
then we take the degree of the corresponding evaluation of $x$, i.e.,
$\alpha_{1}$, as $d_{1}$, and the degree of the corresponding evaluation
of $y$, i.e., $\alpha_{2}$, as $d_{2}$. In this case, the coefficient
of $\alpha^{d_{2}}$ is $P_{3}(\alpha_{3},\dots,\alpha_{v})=\det(M_{1,1})$.
Next, we take $M_{1,1}$, and repeat the same procedure. We do so
until we reach a monomial of degree zero. In this procedure since
we visit all the monomials of $A(x)B(x,y)$ evaluated at different
evaluation points, i.e., $\alpha_{i}$'s, the sum $d_{1}+d_{2}+\dots+d_{v}$
becomes the sum of degrees of all the monomials of $A(x)B(x,y)$.
The next lemma helps us in computing this. 
\begin{lem}
\label{lem:gauss-trick}Consider the polynomial $P(x,y)=\sum_{i=0}^{a}\sum_{j=0}^{b}c_{ij}x^{i}y^{j}$,
where $c_{i,j}$'s are scalars. The sum of the degrees of all the monomials
of $P(x,y)$ is given by $\xi(a,b)\triangleq\frac{a(a+1)}{2}(b+1)+\frac{b(b+1)}{2}(a+1)$. 
\end{lem}
\begin{proof}
The sum of the degrees of all the monomials are given by 
\begin{multline}
\sum_{i=0}^{a}\sum_{j=0}^{b}(i+j)=\sum_{i=0}^{a}i(b+1)+\sum_{i=0}^{a}\sum_{j=0}^{b}j
=\frac{a(a+1)}{2}(b+1)+\frac{b(b+1)}{2}(a+1).
\end{multline}
\end{proof}
By using \lemref{gauss-trick}, the sum of monomial degrees in the
diagonally shaded rectangle in \figref{poly-coeffs} is 
\begin{align}
\xi(K+T-1,L-1)&=\frac{(K+T-1)(K+T)}{2}L+\frac{(L-1)L}{2}(K+T)\\
&= \frac{(K+T)L}{2}(K+L+T-2).
\end{align}
The sum of monomial degrees in the rectangle shaded by crosshatches
in \figref{poly-coeffs} is given by
\begin{align}
\xi&(2K+2T-2,m-1)-\xi(K+T-1,m-1)\nonumber\\
=& \frac{(2K+2T-2)(2K+2T-1)}{2}m+\frac{(m-1)m}{2}(2K+2T-1)\nonumber\\
&-\frac{(K+T-1)(K+T)}{2}m-\frac{(m-1)m}{2}(K+T)\nonumber\\
=& \frac{m}{2}\left(3(K+T)^{2}+m(K+T)-8K-6T-m+3\right).
\end{align}

By summing them we obtain $d_{1}+d_{2}+\dots+d_{v}=$ $\frac{m}{2}\left(3(K+T)^{2}+m(K+T)-8K-6T-m+3\right)$
$+\frac{(K+T)L}{2}(K+L+T-2)$, which concludes the proof.


\vspace{-5pt}
\section{Conclusion}

In this work, for straggler exploitation in SDMM, we have proposed storage- and upload-cost-efficient bivariate Hermitian polynomial codes named SBP codes. Although the previous works usually assume the availability of at least as many workers as the recovery threshold, the multi-message approach allows the completion of the task even if the number of workers is less than the recovery threshold. Compared to univariate polynomial coding based approaches including MM-GASP codes, SBP coding scheme has a lower upload cost and less storage requirement, making the assignment of several sub-tasks to each worker more resource efficient. Thanks to these properties, SBP codes improve the average computation time for SDMM, especially when the number of workers, the upload cost budget, or the storage capacity is limited.
\vspace{-5pt}
\bibliographystyle{IEEEtran}
\bibliography{refs}

\end{document}

%% file: figures/r_th_vs_m_t30.tex
\begin{tikzpicture}[scale=0.75]

\definecolor{color0}{rgb}{0.12156862745098,0.466666666666667,0.705882352941177}
\definecolor{color1}{rgb}{1,0.498039215686275,0.0549019607843137}

\begin{axis}[
legend cell align={left},
legend style={column sep=1pt, fill opacity=0.8, draw opacity=1, text opacity=1, at={(0.97,0.03)}, anchor=south east, draw=white!80!black},
tick align=outside,
tick pos=left,
x grid style={white!69.0196078431373!black},
xlabel={\normalsize Number of computations per worker (m)},
xmin=-3.9, xmax=103.9,
xtick style={color=black},
y grid style={white!69.0196078431373!black},
ylabel={\normalsize $R_{th}$},
ymin=10386.4, ymax=26679.6,
ytick style={color=black},
grid
]
\addplot [ultra thick, densely dashed, color1]
table {%
1 11127
2 13857
3 18387
4 20239
5 20299
6 20359
7 20419
8 20479
9 20539
10 20599
11 20659
12 20719
13 20779
14 20839
15 20899
16 20959
17 21019
18 21079
19 21139
20 21199
21 21259
22 21319
23 21379
24 21439
25 21499
26 21559
27 21619
28 21679
29 21739
30 21799
31 21859
32 21919
33 21979
34 22039
35 22099
36 22159
37 22219
38 22279
39 22339
40 22399
41 22459
42 22519
43 22579
44 22639
45 22699
46 22759
47 22819
48 22879
49 22939
50 22999
51 23059
52 23119
53 23179
54 23239
55 23299
56 23359
57 23419
58 23479
59 23539
60 23599
61 23659
62 23719
63 23779
64 23839
65 23899
66 23959
67 24019
68 24079
69 24139
70 24199
71 24259
72 24319
73 24379
74 24439
75 24499
76 24559
77 24619
78 24679
79 24739
80 24799
81 24859
82 24919
83 24979
84 25039
85 25099
86 25159
87 25219
88 25279
89 25339
90 25399
91 25459
92 25519
93 25579
94 25639
95 25699
96 25759
97 25819
98 25879
99 25939
};
\addlegendentry{\normalsize \arxiv MM-GASP}
\addplot [very thick, color0]
table {%
1 13129
2 13258
3 13387
4 13516
5 13645
6 13774
7 13903
8 14032
9 14161
10 14290
11 14419
12 14548
13 14677
14 14806
15 14935
16 15064
17 15193
18 15322
19 15451
20 15580
21 15709
22 15838
23 15967
24 16096
25 16225
26 16354
27 16483
28 16612
29 16741
30 16870
31 16999
32 17128
33 17257
34 17386
35 17515
36 17644
37 17773
38 17902
39 18031
40 18160
41 18289
42 18418
43 18547
44 18676
45 18805
46 18934
47 19063
48 19192
49 19321
50 19450
51 19579
52 19708
53 19837
54 19966
55 20095
56 20224
57 20353
58 20482
59 20611
60 20740
61 20869
62 20998
63 21127
64 21256
65 21385
66 21514
67 21643
68 21772
69 21901
70 22030
71 22159
72 22288
73 22417
74 22546
75 22675
76 22804
77 22933
78 23062
79 23191
80 23320
81 23449
82 23578
83 23707
84 23836
85 23965
86 24094
87 24223
88 24352
89 24481
90 24610
91 24739
92 24868
93 24997
94 25126
95 25255
96 25384
97 25513
98 25642
99 25771
};
\addlegendentry{\normalsize \arxiv \proposed}

\end{axis}

\end{tikzpicture}

%% file: figures/heter_stable.tex
\begin{tikzpicture}

\definecolor{color0}{rgb}{0.12156862745098,0.466666666666667,0.705882352941177}
\definecolor{color1}{rgb}{1,0.498039215686275,0.0549019607843137}
\definecolor{color2}{rgb}{0.172549019607843,0.627450980392157,0.172549019607843}
\definecolor{color3}{rgb}{0.83921568627451,0.152941176470588,0.156862745098039}
\definecolor{color4}{rgb}{0.580392156862745,0.403921568627451,0.741176470588235}
\definecolor{color5}{rgb}{0.549019607843137,0.337254901960784,0.294117647058824}

\begin{axis}[
legend cell align={left},
legend style={fill opacity=0.8, draw opacity=1, text opacity=1, draw=white!80!black},
log basis x={10},
log basis y={10},
tick align=outside,
tick pos=left,
x grid style={white!69.0196078431373!black},
xlabel={\normalsize Total upload cost budget (UCB) (matrix partitions)},
ylabel={\normalsize Average computation time (ACT) (seconds)},
xmin=19370.062334964, xmax=522713.857338695,
xmode=log,
xtick style={color=black},
xtick={20000,30000,40000,50000,60000,70000,80000,90000,100000,200000,300000,400000,500000},
xticklabels={\(\displaystyle {2\cdot 10^4}\),\(\displaystyle {}\),\(\displaystyle {}\),\(\displaystyle {}\),\(\displaystyle {}\),\(\displaystyle {}\),\(\displaystyle {}\),\(\displaystyle {}\),\(\displaystyle {\normalsize 1\cdot 10^5}\),\(\displaystyle {}\),\(\displaystyle {}\),\(\displaystyle {}\),\(\displaystyle {\normalsize 4\cdot 10^5}\),},
y grid style={white!69.0196078431373!black},
ymin=0.516119507973762, ymax=6.61479145563243,
ymode=log,
ytick style={color=black},
ytick={0.6,0.7,0.8,0.9,1,2,3,4,5},
yticklabels={\(\displaystyle {0.6}\),\(\displaystyle {}\),\(\displaystyle {}\),\(\displaystyle {}\),\(\displaystyle {10^{0}}\),\(\displaystyle {}\),\(\displaystyle {}\),\(\displaystyle {}\),\(\displaystyle {5}\),},
grid
]
\addplot [very thick, color3, mark=diamond*, mark size=2, mark options={solid}]
table {%
135000 1.68285389063249
180000 1.68253954228528
225000 1.68285389063249
270000 1.68285389063249
315000 1.68285389063249
360000 1.68285389063249
405000 1.68285389063249
450000 1.68285389063249
};
\addlegendentry{\normalsize \arxiv Rateless, c=5}
\addplot [very thick, color5, mark=pentagon*, mark size=2, mark options={solid}]
table {%
180000 1.49655793589903
225000 1.50176343134698
270000 1.49655793589903
315000 1.49655793589903
360000 1.49655793589903
405000 1.49655793589903
450000 1.49655793589903
};
\addlegendentry{\normalsize \arxiv Rateless, c=7}
\addplot [very thick, color0, mark=*, mark size=2, mark options={solid}]
table {%
22500 4.92301699149631
27000 2.9722230615861
31500 1.95500082962587
36000 1.65269501530846
40500 1.36229747117828
45000 1.12120289524777
54000 1.11240248564855
63000 1.11240248564855
72000 1.11240248564855
90000 1.11240248564855
135000 1.11240248564855
180000 1.11240248564855
225000 1.11240248564855
270000 1.11240248564855
315000 1.11240248564855
360000 1.11240248564855
405000 1.11240248564855
450000 1.11240248564855
};
\addlegendentry{\normalsize \arxiv \proposed}
\addplot [very thick, color2, mark=x, mark size=3, mark options={solid}]
table {%
72000 1.0101708248489
90000 1.01347238203575
135000 1.0101708248489
180000 1.0101708248489
225000 1.0101708248489
270000 1.0101708248489
315000 1.0101708248489
360000 1.0101708248489
405000 1.0101708248489
450000 1.0101708248489
};
\addlegendentry{\normalsize \arxiv Rateless, c=3}
\addplot [very thick, color4, mark=triangle, mark size=3, mark options={solid}]
table {%
135000 0.984732478261053
180000 0.983995814349684
225000 0.984732478261053
270000 0.984732478261053
315000 0.984732478261053
360000 0.984732478261053
405000 0.984732478261053
450000 0.984732478261053
};
\addlegendentry{\normalsize \arxiv Rateless, c=6}
\addplot [very thick, color1, mark=square*, mark size=2, mark options={solid}]
table {%
54000 5.88684237668578
63000 3.27680255866737
72000 2.07836106157089
90000 1.12562074968045
135000 0.654368164634346
180000 0.605500951465565
225000 0.578638095330492
270000 0.580545487165599
315000 0.626019983649504
360000 0.671098373537379
405000 0.716801077289796
450000 0.762915812241203
};
\addlegendentry{\normalsize \arxiv MM-GASP}

\end{axis}

\end{tikzpicture}

%% file: figures/r_th_hetero_stable.tex
\begin{tikzpicture}

\definecolor{color0}{rgb}{0.12156862745098,0.466666666666667,0.705882352941177}
\definecolor{color1}{rgb}{1,0.498039215686275,0.0549019607843137}
\definecolor{color2}{rgb}{0.172549019607843,0.627450980392157,0.172549019607843}
\definecolor{color3}{rgb}{0.83921568627451,0.152941176470588,0.156862745098039}
\definecolor{color4}{rgb}{0.580392156862745,0.403921568627451,0.741176470588235}
\definecolor{color5}{rgb}{0.549019607843137,0.337254901960784,0.294117647058824}

\begin{axis}[
legend cell align={left},
legend style={fill opacity=0.8, draw opacity=1, text opacity=1, draw=white!80!black, at={(0.55,0.1)},anchor=west},
legend columns=2,
tick align=outside,
tick pos=left,
xlabel={\normalsize Total upload cost budget (UCB) (matrix partitions)},
ylabel={\normalsize $R_{th}$},
x grid style={white!69.0196078431373!black},
xmin=1125, xmax=471375,
xtick style={color=black},
scaled x ticks = false,
scaled y ticks = true,
y grid style={white!69.0196078431373!black},
ymin=6000, ymax=81320.2425,
ytick style={color=black},
grid
]

\addplot [very thick, color5, mark=pentagon*, mark size=2, mark options={solid}]
table {%
180000 77422.02
225000 77422.02
270000 77422.02
315000 77422.02
360000 77422.02
405000 77422.02
450000 77422.02
};
\addlegendentry{\normalsize \arxiv Rateless, c=7}
\addplot [very thick, color4, mark=triangle, mark size=3, mark options={solid}]
table {%
135000 57350.25
180000 57350.25
225000 57350.25
270000 57350.25
315000 57350.25
360000 57350.25
405000 57350.25
450000 57350.25
};
\addlegendentry{\normalsize \arxiv Rateless, c=6}
\addplot [very thick, color3, mark=diamond*, mark size=2, mark options={solid}]
table {%
135000 48267
180000 48267
225000 48267
270000 48267
315000 48267
360000 48267
405000 48267
450000 48267
};
\addlegendentry{\normalsize \arxiv Rateless, c=5}
\addplot [very thick, color1, mark=square*, mark size=2, mark options={solid}]
table {%
54000 23599
63000 24199
72000 24799
90000 25999
135000 28999
180000 31999
225000 34999
270000 37999
315000 40999
360000 43999
405000 46999
450000 49999
};
\addlegendentry{\normalsize \arxiv MM-GASP}
\addplot [very thick, color2, mark=x, mark size=3, mark options={solid}]
table {%
72000 34516.5
90000 34516.5
135000 34516.5
180000 34516.5
225000 34516.5
270000 34516.5
315000 34516.5
360000 34516.5
405000 34516.5
450000 34516.5
};
\addlegendentry{\normalsize \arxiv Rateless, c=3}
\addplot [very thick, color0, mark=*, mark size=2, mark options={solid}]
table {%
22500 19321
27000 20611
31500 21901
36000 23191
40500 24481
45000 25771
54000 25900
63000 25900
72000 25900
90000 25900
135000 25900
180000 25900
225000 25900
270000 25900
315000 25900
360000 25900
405000 25900
450000 25900
};
\addlegendentry{\normalsize \arxiv \proposed}

\end{axis}

\end{tikzpicture}

%% file: figures/heter_instable.tex
\begin{tikzpicture}

\definecolor{color0}{rgb}{0.12156862745098,0.466666666666667,0.705882352941177}
\definecolor{color1}{rgb}{1,0.498039215686275,0.0549019607843137}
\definecolor{color2}{rgb}{0.172549019607843,0.627450980392157,0.172549019607843}
\definecolor{color3}{rgb}{0.83921568627451,0.152941176470588,0.156862745098039}
\definecolor{color4}{rgb}{0.580392156862745,0.403921568627451,0.741176470588235}
\definecolor{color5}{rgb}{0.549019607843137,0.337254901960784,0.294117647058824}

\begin{axis}[
legend cell align={left},
legend style={fill opacity=0.8, draw opacity=1, text opacity=1, draw=white!80!black,at={(0.55,0.9)},anchor=west},
legend columns=2,
log basis x={10},
log basis y={10},
tick align=outside,
tick pos=left,
x grid style={white!69.0196078431373!black},
xlabel={\normalsize Total upload cost budget (UCB) (matrix partitions)},
ylabel={\normalsize Average computation time (ACT) (seconds)},
xmin=19370.062334964, xmax=522713.857338695,
xmode=log,
xtick style={color=black},
xtick={20000,30000,40000,50000,60000,70000,80000,90000,100000,200000,300000,400000,500000},
xticklabels={\(\displaystyle {2\cdot 10^4}\),\(\displaystyle {}\),\(\displaystyle {}\),\(\displaystyle {}\),\(\displaystyle {}\),\(\displaystyle {}\),\(\displaystyle {}\),\(\displaystyle {}\),\(\displaystyle {\normalsize1\cdot 10^5}\),\(\displaystyle {}\),\(\displaystyle {}\),\(\displaystyle {}\),\(\displaystyle {\normalsize 4\cdot 10^5}\),},
y grid style={white!69.0196078431373!black},
ymin=0.516119507973762, ymax=6.61479145563243,
ymode=log,
ytick style={color=black},
ytick={0.6,0.7,0.8,0.9,1,2,3,4,5},
yticklabels={\(\displaystyle {0.6}\),\(\displaystyle {}\),\(\displaystyle {}\),\(\displaystyle {}\),\(\displaystyle {10^{0}}\),\(\displaystyle {}\),\(\displaystyle {}\),\(\displaystyle {}\),\(\displaystyle {5}\),},
grid
]
\addplot [very thick, color5, mark=pentagon*, mark size=2, mark options={solid}]
table {%
180000 3.35844822666207
225000 3.24897571909394
270000 3.25481990729134
315000 3.24897571909394
360000 3.24897571909394
405000 3.24897571909394
450000 3.24897571909394

};
\addlegendentry{\normalsize \arxiv Rateless, c=7}
\addplot [very thick, color3, mark=diamond*, mark size=2, mark options={solid}]
table {%
135000 3.28421984567467
180000 3.27623931678182
225000 3.28421984567467
270000 3.28421984567467
315000 3.28421984567467
360000 3.28421984567467
405000 3.28421984567467
450000 3.28421984567467
};
\addlegendentry{\normalsize \arxiv Rateless, c=5}
\addplot [very thick, color2, mark=x, mark size=3, mark options={solid}]
table {%
72000 3.07573119635437
90000 3.08459522144747
135000 3.07573119635437
180000 3.07573119635437
225000 3.07573119635437
270000 3.07573119635437
315000 3.07573119635437
360000 3.07573119635437
405000 3.07573119635437
450000 3.07573119635437
};
\addlegendentry{\normalsize \arxiv Rateless, c=3}

\addplot [very thick, color4, mark=triangle, mark size=3, mark options={solid}]
table {%
135000 2.8500594703337
180000 2.76561373748261
225000 2.75347746450472
270000 2.76561373748261
315000 2.76561373748261
360000 2.76561373748261
405000 2.76561373748261
450000 2.76561373748261
};
\addlegendentry{\normalsize \arxiv Rateless, c=6}
\addplot [very thick, color0, mark=*, mark size=2, mark options={solid}]
table {%
22500 4.93108040137729
27000 2.96566253977127
31500 1.95684858838862
36000 1.65559995616741
40500 1.3626600846039
45000 1.12102161692354
54000 1.11024703820886
63000 1.11024703820886
72000 1.11024703820886
90000 1.11024703820886
135000 1.11024703820886
180000 1.11024703820886
225000 1.11024703820886
270000 1.11024703820886
315000 1.11024703820886
360000 1.11024703820886
405000 1.11024703820886
450000 1.11024703820886
};
\addlegendentry{\normalsize \arxiv \proposed}
\addplot [very thick, color1, mark=square*, mark size=2, mark options={solid}]
table {%
54000 5.88669805836413
63000 3.2680315819786
72000 2.07429370663684
90000 1.12748639436715
135000 0.655047204711411
180000 0.606029423176481
225000 0.578439469686378
270000 0.580629458773439
315000 0.625896885928243
360000 0.671290469885768
405000 0.717007620276599
450000 0.762871224303773
};
\addlegendentry{\normalsize \arxiv MM-GASP}

\end{axis}

\end{tikzpicture}

%% file: figures/homo_stable.tex
\begin{tikzpicture}

\definecolor{color0}{rgb}{0.12156862745098,0.466666666666667,0.705882352941177}
\definecolor{color1}{rgb}{1,0.498039215686275,0.0549019607843137}
\definecolor{color2}{rgb}{0.172549019607843,0.627450980392157,0.172549019607843}
\definecolor{color3}{rgb}{0.83921568627451,0.152941176470588,0.156862745098039}
\definecolor{color4}{rgb}{0.580392156862745,0.403921568627451,0.741176470588235}
\definecolor{color5}{rgb}{0.549019607843137,0.337254901960784,0.294117647058824}

\begin{axis}[
legend cell align={left},
legend style={fill opacity=0.8, draw opacity=1, text opacity=1, draw=white!80!black,at={(0.02,0.88)},anchor=west},
legend columns=1,
log basis x={10},
log basis y={10},
tick align=outside,
tick pos=left,
x grid style={white!69.0196078431373!black},
xlabel={\normalsize Total upload cost budget (UCB) (matrix partitions)},
ylabel={\normalsize Average computation time (ACT) (seconds)},
xmin=19370.062334964, xmax=522713.857338695,
xmode=log,
xtick style={color=black},
xtick={20000,30000,40000,50000,60000,70000,80000,90000,100000,200000,300000,400000,500000},
xticklabels={\(\displaystyle {2\cdot 10^4}\),\(\displaystyle {}\),\(\displaystyle {}\),\(\displaystyle {}\),\(\displaystyle {}\),\(\displaystyle {}\),\(\displaystyle {}\),\(\displaystyle {}\),\(\displaystyle {\normalsize1 \cdot 10^5}\),\(\displaystyle {}\),\(\displaystyle {}\),\(\displaystyle {}\),\(\displaystyle {\normalsize 4\cdot 10^5}\),},
y grid style={white!69.0196078431373!black},
ymin=0.458187437299291, ymax=1.28154411705036,
ymode=log,
ytick style={color=black},
ytick={0.5,0.6,0.7,0.8,0.9,1},
yticklabels={\(\displaystyle {}\),\(\displaystyle {0.6}\),\(\displaystyle {}\),\(\displaystyle {}\),\(\displaystyle {}\),\(\displaystyle {10^{0}}\)},
grid
]
\addplot [very thick, color1, mark=square*, mark size=2, mark options={solid}]
table {%
54000 0.583593840946327
63000 0.593029512157708
72000 0.607047226537189
90000 0.636346681050562
135000 0.710139496019787
180000 0.783164158715691
225000 0.856593552014248
270000 0.929570264259294
315000 1.00255488836052
360000 1.07620215855886
405000 1.14977186521142
450000 1.22315870269053
};
\addlegendentry{\normalsize \arxiv MM-GASP}

\addplot [very thick, color4, mark=triangle, mark size=3, mark options={solid}]
table {%
63000 0.939662935839085
72000 0.937977304011021
90000 0.939662935839085
135000 0.939662935839085
180000 0.939662935839085
225000 0.939662935839085
270000 0.939662935839085
315000 0.939662935839085
360000 0.939662935839085
405000 0.939662935839085
450000 0.939662935839085
};
\addlegendentry{\normalsize \arxiv Rateless, c=3}

\addplot [very thick, color3, mark=diamond*, mark size=2, mark options={solid}]
table {%
54000 0.889415655522552
63000 0.878449805545752
72000 0.87934068696157
90000 0.878449805545752
135000 0.878449805545752
180000 0.878449805545752
225000 0.878449805545752
270000 0.878449805545752
315000 0.878449805545752
360000 0.878449805545752
405000 0.878449805545752
450000 0.878449805545752
};
\addlegendentry{\normalsize \arxiv Rateless, c=2}
\addplot [very thick, color2, mark=x, mark size=3, mark options={solid}]
table {%
54000 0.846740337530054
63000 0.840348559774894
72000 0.846740337530054
90000 0.846740337530054
135000 0.846740337530054
180000 0.846740337530054
225000 0.846740337530054
270000 0.846740337530054
315000 0.846740337530054
360000 0.846740337530054
405000 0.846740337530054
450000 0.846740337530054
};
\addlegendentry{\normalsize \arxiv Rateless, c=1}

\addplot [very thick, color0, mark=*, mark size=2, mark options={solid}]
table {%
22500 0.479596601295886
27000 0.505364729689191
31500 0.536935788736443
36000 0.567785135062818
40500 0.599738981714983
45000 0.631303084123025
54000 0.633983846349865
63000 0.633983846349865
72000 0.633983846349865
90000 0.633983846349865
135000 0.633983846349865
180000 0.633983846349865
225000 0.633983846349865
270000 0.633983846349865
315000 0.633983846349865
360000 0.633983846349865
405000 0.633983846349865
450000 0.633983846349865
};
\addlegendentry{\normalsize \arxiv \proposed}

\addplot [very thick, color0, dashed]
table {%
22500 0.479596601295886
450000 0.479596601295886
};

\addplot [very thick, color1, dashed]
table {%
54000 0.583593840946327
450000 0.583593840946327
};

\end{axis}

\end{tikzpicture}

%% file: figures/homo_mostly_stable.tex
\begin{tikzpicture}

\definecolor{color0}{rgb}{0.12156862745098,0.466666666666667,0.705882352941177}
\definecolor{color1}{rgb}{1,0.498039215686275,0.0549019607843137}
\definecolor{color2}{rgb}{0.172549019607843,0.627450980392157,0.172549019607843}
\definecolor{color3}{rgb}{0.83921568627451,0.152941176470588,0.156862745098039}
\definecolor{color4}{rgb}{0.580392156862745,0.403921568627451,0.741176470588235}
\definecolor{color5}{rgb}{0.549019607843137,0.337254901960784,0.294117647058824}

\begin{axis}[
legend cell align={left},
legend style={fill opacity=0.8, draw opacity=1, text opacity=1, draw=white!80!black,at={(0.02,0.88)},anchor=west},
legend columns=1,
log basis x={10},
log basis y={10},
tick align=outside,
tick pos=left,
x grid style={white!69.0196078431373!black},
xlabel={\normalsize Total upload cost budget (UCB) (matrix partitions)},
ylabel={\normalsize Average computation time (ACT) (seconds)},
xmin=19370.062334964, xmax=522713.857338695,
xmode=log,
xtick style={color=black},
xtick={20000,30000,40000,50000,60000,70000,80000,90000,100000,200000,300000,400000,500000},
xticklabels={\(\displaystyle {2\cdot 10^4}\),\(\displaystyle {}\),\(\displaystyle {}\),\(\displaystyle {}\),\(\displaystyle {}\),\(\displaystyle {}\),\(\displaystyle {}\),\(\displaystyle {}\),\(\displaystyle {\normalsize 1\cdot 10^5}\),\(\displaystyle {}\),\(\displaystyle {}\),\(\displaystyle {}\),\(\displaystyle {\normalsize 4\cdot 10^5}\),},
y grid style={white!69.0196078431373!black},
ymin=0.412786611662386, ymax=11.2754730436651,
ymode=log,
ytick style={color=black},
ytick={0.5,0.6,0.7,0.8,0.9,1,2,3,4,5,6,7,8,9,10},
yticklabels={\(\displaystyle {0.5}\),\(\displaystyle {}\),\(\displaystyle {}\),\(\displaystyle {}\),\(\displaystyle {}\),\(\displaystyle {10^{0}}\),\(\displaystyle {}\),\(\displaystyle {}\),\(\displaystyle {}\),\(\displaystyle {}\),\(\displaystyle {}\),\(\displaystyle {}\),\(\displaystyle {}\),\(\displaystyle {}\),\(\displaystyle {10^1}\)},
grid
]

\addplot [very thick, color2, mark=x, mark size=3, mark options={solid}]
table {%
54000 9.70160180156962
63000 9.67499338334633
72000 9.70160180156962
90000 9.70160180156962
135000 9.70160180156962
180000 9.70160180156962
225000 9.70160180156962
270000 9.70160180156962
315000 9.70160180156962
360000 9.70160180156962
405000 9.70160180156962
450000 9.70160180156962
};
\addlegendentry{\normalsize \arxiv Rateless, c=1}
\addplot [very thick, color3, mark=diamond*, mark size=2, mark options={solid}]
table {%
63000 8.40990960251967
72000 8.41929091683649
90000 8.40990960251967
135000 8.40990960251967
180000 8.40990960251967
225000 8.40990960251967
270000 8.40990960251967
315000 8.40990960251967
360000 8.40990960251967
405000 8.40990960251967
450000 8.40990960251967
};
\addlegendentry{\normalsize \arxiv Rateless, c=2}
\addplot [very thick, color4, mark=triangle, mark size=3, mark options={solid}]
table {%
63000 7.88515677910686
72000 7.84660625037254
90000 7.88515677910686
135000 7.88515677910686
180000 7.88515677910686
225000 7.88515677910686
270000 7.88515677910686
315000 7.88515677910686
360000 7.88515677910686
405000 7.88515677910686
450000 7.88515677910686
};
\addlegendentry{\normalsize \arxiv Rateless, c=3}

\addplot [very thick, color1, mark=square*, mark size=2, mark options={solid}]
table {%
54000 0.583063974842608
63000 0.592578464843454
72000 0.606989123483907
90000 0.636508999170937
135000 0.710225204014517
180000 0.783046996454747
225000 0.856066174854733
270000 0.929353834275177
315000 1.00301630289755
360000 1.07587511799463
405000 1.14948652610197
450000 1.22307798619296
};
\addlegendentry{\normalsize \arxiv MM-GASP}

\addplot [very thick, color0, mark=*, mark size=2, mark options={solid}]
table {%
22500 0.479752148952563
27000 0.505563605200671
31500 0.535872860087322
36000 0.568318098317474
40500 0.598917364334936
45000 0.630206337907349
54000 0.634151791293423
63000 0.634151791293423
72000 0.634151791293423
90000 0.634151791293423
135000 0.634151791293423
180000 0.634151791293423
225000 0.634151791293423
270000 0.634151791293423
315000 0.634151791293423
360000 0.634151791293423
405000 0.634151791293423
450000 0.634151791293423
};
\addlegendentry{\normalsize \arxiv \proposed}

\addplot [very thick, color0, dashed]
table {%
22500 0.479752148952563
450000 0.479752148952563
};

\addplot [very thick, color1, dashed]
table {%
54000 0.583063974842608
450000 0.583063974842608
};

\end{axis}

\end{tikzpicture}

%% file: figures/degxdegy.tex
\usetikzlibrary{patterns} \begin{tikzpicture}[scale=1.8]
\draw[very thin]  (-5,2.5) rectangle (-2,0.5); 
\draw[very thin]  (-2,1.5) rectangle (0.5,0.5);
\path [pattern=north east lines, pattern color = black!20]  (-5,2.5) rectangle (-2,0.5); 
\path [pattern=crosshatch, pattern color = black!20]  (-2,1.5) rectangle (0.5,0.5);
\draw[fill, color=black] (-4.5,0.5) node (v1) {} circle (.05);
\draw[fill, color=black] (-4.5,2.5) node (v1) {} circle (.05);
\draw[fill, color=black] (-4.5,2) node (v1) {} circle (.05);
\draw[fill, color=black] (-4.5,1.5) node (v1) {} circle (.05);
\draw[fill, color=black] (-4.5,1) node (v1) {} circle (.05);
\draw[fill, color=black] (-5,0.5) node (v0) {} circle (.05); 
\draw[fill, color=black] (-5,2.5) node (v1) {} circle (.05);
\draw[fill, color=black] (-5,2) node (v1) {} circle (.05);
\draw[fill, color=black] (-5,1.5) node (v1) {} circle (.05);
\draw[fill, color=black] (-5,1) node (v1) {} circle (.05);
\draw[fill, color=black] (-4,0.5) node (v1) {} circle (.05);
\draw[fill, color=black] (-4,2.5) node (v1) {} circle (.05);
\draw[fill, color=black] (-4,2) node (v1) {} circle (.05);
\draw[fill, color=black] (-4,1.5) node (v1) {} circle (.05);
\draw[fill, color=black] (-4,1) node (v1) {} circle (.05);
\draw[fill, color=black] (-2,0.5) node (v1) {} circle (.05);
\draw[fill, color=black] (-2,2.5) node (v1) {} circle (.05);
\draw[fill, color=black] (-2,2) node (v1) {} circle (.05);
\draw[fill, color=black] (-2,1.5) node (v1) {} circle (.05);
\draw[fill, color=black] (-2,1) node (v1) {} circle (.05);
\draw[fill, color=black] (-2.5,0.5) node (v1) {} circle (.05);
\draw[fill, color=black] (-2.5,2.5) node (v1) {} circle (.05);
\draw[fill, color=black] (-2.5,2) node (v1) {} circle (.05);
\draw[fill, color=black] (-2.5,1.5) node (v1) {} circle (.05);
\draw[fill, color=black] (-2.5,1) node (v1) {} circle (.05);
\draw[fill, color=black] (-1.5,1.5) node (v1) {} circle (.05); 
\draw[fill, color=black] (-1.5,0.5) node (v1) {} circle (.05); 
\draw[fill, color=black] (-1.5,1) node (v1) {} circle (.05);
\draw[fill, color=black] (0,1.5) node (v1) {} circle (.05); 
\draw[fill, color=black] (0,0.5) node (v1) {} circle (.05); 
\draw[fill, color=black] (0,1) node (v1) {} circle (.05);
\draw[fill, color=black] (0.5,1.5) node (v1) {} circle (.05); 
\draw[fill, color=black] (0.5,0.5) node (v1) {} circle (.05); 
\draw[fill, color=black] (0.5,1) node (v1) {} circle (.05);
\node at (-0.75,1) {$\cdots$}; 
\node at (-3.25,1) {$\cdots$}; 
\node at (-3.25,1.5) {$\cdots$};
\node at (-3.25,2) {$\cdots$};
\node[scale=0.9] at (-5,0.25) {$0$}; 
\node[scale=0.9] at (-5.25,0.5) {$0$};
\node[scale=0.9] at (-2,0.25) {$K+T-1$};
\node[scale=0.9] at (-5.375,1.5) {$m-1$}; 
\node[scale=0.9] at (-5.375,2.5) {$L-1$}; 
\node[scale=0.9] at (0.5,0.25) {$2K+2T-2$};

\node[scale=0.9] (v2) at (1,0.5) {$\deg(x)$};
\draw [-latex] (v0) edge (v2);
\node[scale=0.9] (v3) at (-5,3) {$\deg(y)$};
\draw [-latex] (v0) edge (v3);

\end{tikzpicture}

%% file: figures/coalescence_ex1_piv.tex
\begin{tikzpicture}[scale=1.3]
\draw[very thin]  (-5,2.5) rectangle (-2.5,0.5); 
\draw[very thin]  (-2.5,1.5) rectangle (0,0.5);
\draw [help lines,  step=0.5cm] (-5,0.5) node (v19) {} grid (-2.5,2.5); 
\draw [help lines,  step=0.5cm] (-2.5,0.5) node (v19) {} grid (0,1.5); 
\draw[fill, color=black] (-4.5,0.5) node (v1) {} circle (.08);
\draw[color=black] (-4.5,2.5) node (v1) {} circle (.08);
\draw[color=black] (-4.5,2) node (v1) {} circle (.08);
\draw[fill, color=black] (-4.5,1.5) node (v1) {} circle (.08);
\draw[fill, color=black] (-4.5,1) node (v1) {} circle (.08);
\draw[fill, color=black] (-5,0.5) node (v0) {} circle (.08); 
\draw[fill, color=black] (-5,2.5) node (v1) {} circle (.08);
\draw[fill, color=black] (-5,2) node (v1) {} circle (.08);
\draw[fill, color=black] (-5,1.5) node (v1) {} circle (.08);
\draw[fill, color=black] (-5,1) node (v1) {} circle (.08);
\draw[color=black] (-4,0.5) node (v1) {} circle (.08);
\draw[color=black] (-4,2.5) node (v1) {} circle (.08);
\draw[color=black] (-4,2) node (v1) {} circle (.08);
\draw[color=black] (-4,1.5) node (v1) {} circle (.08);
\draw[color=black] (-4,1) node (v1) {} circle (.08);
\draw[color=black] (-2.5,0.5) node (v1) {} circle (.08);
\draw[color=black] (-1,1.5) node (v1) {} circle (.08);
\draw[color=black] (-2,0.5) node (v1) {} circle (.08);
\draw[color=black] (-2,1.5) node (v1) {} circle (.08);
\draw[color=black] (-2,1) node (v1) {} circle (.08);
\draw[color=black] (-1,1) node (v1) {} circle (.08);
\draw[color=black] (-2.5,2.5) node (v1) {} circle (.08);
\draw[color=black] (-2.5,2) node (v1) {} circle (.08);
\draw[color=black] (-2.5,1.5) node (v1) {} circle (.08);
\draw[color=black] (-2.5,1) node (v1) {} circle (.08);
\draw[color=black] (-1.5,1.5) node (v1) {} circle (.08); 
\draw[color=black] (-1.5,0.5) node (v1) {} circle (.08); 
\draw[color=black] (-1.5,1) node (v1) {} circle (.08);
\draw[color=black] (0,1.5) node (v1) {} circle (.08); 
\draw[color=black] (0,0.5) node (v1) {} circle (.08); 
\draw[color=black] (0,1) node (v1) {} circle (.08);
\draw[color=black] (-0.5,1.5) node (v1) {} circle (.08); 
\draw[color=black] (-0.5,0.5) node (v1) {} circle (.08); 
\draw[color=black] (-0.5,1) node (v1) {} circle (.08);
\draw[color=black] (-3.5,2.5) node (v1) {} circle (.08);
\draw[color=black] (-3,1.5) node (v1) {} circle (.08);
\draw[color=black] (-3,2) node (v1) {} circle (.08);
\draw[color=black] (-3,2.5) node (v1) {} circle (.08);
\draw[color=black] (-3.5,0.5) node (v1) {} circle (.08);
\draw[color=black] (-3.5,1) node (v1) {} circle (.08);
\draw[color=black] (-3.5,1.5) node (v1) {} circle (.08);
\draw[color=black] (-3.5,2) node (v1) {} circle (.08);
\draw[color=black] (-3,0.5) node (v1) {} circle (.08);
\draw[color=black] (-3,1) node (v1) {} circle (.08);
\draw[color=black] (-1,0.5) node (v1) {} circle (.08);

\node[scale=0.7] at (-5,0.25) {$0$}; 
\node[scale=0.7] at (-4.5,0.25) {$1$}; 
\node[scale=0.7] at (-4,0.25) {$2$}; 
\node[scale=0.7] at (-3.5,0.25) {$3$}; 
\node[scale=0.7] at (-3,0.25) {$4$}; 
\node[scale=0.7] at (-5.25,0.5) {$0$};
\node[scale=0.7] at (-5.25,1) {$1$};
\node[scale=0.7] at (-5.25,2) {$3$};
\node[scale=0.7] at (-2.5,0.25) {$5$};
\node[scale=0.7] at (-5.25,1.5) {$2$}; 
\node[scale=0.7] at (-5.25,2.5) {$4$}; 
\node[scale=0.7] at (0,0.25) {$10$};
\node[scale=0.7] at (-0.5,0.25) {$9$};
\node[scale=0.7] at (-1,0.25) {$8$};
\node[scale=0.7] at (-1.5,0.25) {$7$};
\node[scale=0.7] at (-2,0.25) {$6$};

\end{tikzpicture}

%% file: figures/coalescence_ex1_var.tex
\begin{tikzpicture}[scale=1.3]
\definecolor{color0}{rgb}{0.12156862745098,0.466666666666667,0.705882352941177}
\definecolor{color1}{rgb}{1,0.498039215686275,0.0549019607843137}
\definecolor{color2}{rgb}{0.172549019607843,0.627450980392157,0.172549019607843}
\definecolor{color3}{rgb}{0.83921568627451,0.152941176470588,0.156862745098039}

\draw[very thin]  (-5,2.5) rectangle (-2.5,0.5); 
\draw[very thin]  (-2.5,1.5) rectangle (0,0.5);
\draw [help lines,  step=0.5cm] (-5,0.5) node (v19) {} grid (-2.5,2.5); 
\draw [help lines,  step=0.5cm] (-2.5,0.5) node (v19) {} grid (0,1.5); 
\draw[color=black] (-4.5,0.5) node (v1) {} circle (.08);
\draw[color=black] (-4.5,2.5) node (v1) {} circle (.08);
\draw[color=black] (-4.5,2) node (v1) {} circle (.08);
\draw[color=black] (-4.5,1.5) node (v1) {} circle (.08);
\draw[color=black] (-4.5,1) node (v1) {} circle (.08);
\draw[fill, color=color0] (-5,0.5) node (v0) {} circle (.08); 
\draw[color=black] (-5,2.5) node (v1) {} circle (.08);
\draw[color=black] (-5,2) node (v1) {} circle (.08);
\draw[fill, color=color0] (-5,1.5) node (v1) {} circle (.08);
\draw[fill, color=color0] (-5,1) node (v1) {} circle (.08);
\draw[color=black] (-4,0.5) node (v1) {} circle (.08);
\draw[color=black] (-4,2.5) node (v1) {} circle (.08);
\draw[color=black] (-4,2) node (v1) {} circle (.08);
\draw[color=black] (-4,1.5) node (v1) {} circle (.08);
\draw[color=black] (-4,1) node (v1) {} circle (.08);
\draw[color=black] (-2.5,0.5) node (v1) {} circle (.08);
\draw[color=black] (-1,1.5) node (v1) {} circle (.08);
\draw[color=black] (-2,0.5) node (v1) {} circle (.08);
\draw[color=black] (-2,1.5) node (v1) {} circle (.08);
\draw[color=black] (-2,1) node (v1) {} circle (.08);
\draw[color=black] (-1,1) node (v1) {} circle (.08);
\draw[color=black] (-2.5,2.5) node (v1) {} circle (.08);
\draw[color=black] (-2.5,2) node (v1) {} circle (.08);
\draw[color=black] (-2.5,1.5) node (v1) {} circle (.08);
\draw[color=black] (-2.5,1) node (v1) {} circle (.08);
\draw[color=black] (-1.5,1.5) node (v1) {} circle (.08); 
\draw[color=black] (-1.5,0.5) node (v1) {} circle (.08); 
\draw[color=black] (-1.5,1) node (v1) {} circle (.08);
\draw[color=black] (0,1.5) node (v1) {} circle (.08); 
\draw[color=black] (0,0.5) node (v1) {} circle (.08); 
\draw[color=black] (0,1) node (v1) {} circle (.08);
\draw[color=black] (-0.5,1.5) node (v1) {} circle (.08); 
\draw[color=black] (-0.5,0.5) node (v1) {} circle (.08); 
\draw[color=black] (-0.5,1) node (v1) {} circle (.08);
\draw[color=black] (-3.5,2.5) node (v1) {} circle (.08);
\draw[color=black] (-3,1.5) node (v1) {} circle (.08);
\draw[color=black] (-3,2) node (v1) {} circle (.08);
\draw[color=black] (-3,2.5) node (v1) {} circle (.08);
\draw[color=black] (-3.5,0.5) node (v1) {} circle (.08);
\draw[color=black] (-3.5,1) node (v1) {} circle (.08);
\draw[color=black] (-3.5,1.5) node (v1) {} circle (.08);
\draw[color=black] (-3.5,2) node (v1) {} circle (.08);
\draw[color=black] (-3,0.5) node (v1) {} circle (.08);
\draw[color=black] (-3,1) node (v1) {} circle (.08);
\draw[color=black] (-1,0.5) node (v1) {} circle (.08);

\node[scale=0.7] at (-5,0.25) {$0$}; 
\node[scale=0.7] at (-4.5,0.25) {$1$}; 
\node[scale=0.7] at (-4,0.25) {$2$}; 
\node[scale=0.7] at (-3.5,0.25) {$3$}; 
\node[scale=0.7] at (-3,0.25) {$4$}; 
\node[scale=0.7] at (-5.25,0.5) {$0$};
\node[scale=0.7] at (-5.25,1) {$1$};
\node[scale=0.7] at (-5.25,2) {$3$};
\node[scale=0.7] at (-2.5,0.25) {$5$};
\node[scale=0.7] at (-5.25,1.5) {$2$}; 
\node[scale=0.7] at (-5.25,2.5) {$4$}; 
\node[scale=0.7] at (0,0.25) {$10$};
\node[scale=0.7] at (-0.5,0.25) {$9$};
\node[scale=0.7] at (-1,0.25) {$8$};
\node[scale=0.7] at (-1.5,0.25) {$7$};
\node[scale=0.7] at (-2,0.25) {$6$};

\node[scale=0.8] at (-4.875,0.625) {a};
\node[scale=0.8] at (-4.875,1.125) {b};
\node[scale=0.8] at (-4.875,1.625) {c};

\end{tikzpicture}

%% file: figures/coalescence_ex1.tex
\begin{tikzpicture}[scale=1.3]
\definecolor{color0}{rgb}{0.12156862745098,0.466666666666667,0.705882352941177}

\draw[very thin]  (-5,2.5) rectangle (-2.5,0.5); 
\draw[very thin]  (-2.5,1.5) rectangle (0,0.5);
\draw [help lines,  step=0.5cm] (-5,0.5) node (v19) {} grid (-2.5,2.5); 
\draw [help lines,  step=0.5cm] (-2.5,0.5) node (v19) {} grid (0,1.5); 
\draw[fill, color=black] (-4.5,0.5) node (v1) {} circle (.08);
\draw[fill, color=color0] (-4.5,2.5) node (v1) {} circle (.08);
\draw[color=black] (-4,1) node (v1) {} circle (.08);
\draw[fill, color=black] (-4.5,1.5) node (v1) {} circle (.08);
\draw[fill, color=black] (-4.5,1) node (v1) {} circle (.08);
\draw[fill, color=black] (-5,0.5) node (v0) {} circle (.08); 
\draw[fill, color=black] (-5,2.5) node (v1) {} circle (.08);
\draw[fill, color=black] (-5,2) node (v1) {} circle (.08);
\draw[fill, color=black] (-5,1.5) node (v1) {} circle (.08);
\draw[fill, color=black] (-5,1) node (v1) {} circle (.08);
\draw[fill, color=color0] (-4,0.5) node (v1) {} circle (.08);
\draw[color=black] (-4,2.5) node (v1) {} circle (.08);
\draw[color=black] (-4,2) node (v1) {} circle (.08);
\draw[color=black] (-4,1.5) node (v1) {} circle (.08);
\draw[fill, color=color0] (-4.5,2) node (v1) {} circle (.08);
\draw[color=black] (-2.5,0.5) node (v1) {} circle (.08);
\draw[color=black] (-1,1.5) node (v1) {} circle (.08);
\draw[color=black] (-2,0.5) node (v1) {} circle (.08);
\draw[color=black] (-2,1.5) node (v1) {} circle (.08);
\draw[color=black] (-2,1) node (v1) {} circle (.08);
\draw[color=black] (-1,1) node (v1) {} circle (.08);
\draw[color=black] (-2.5,2.5) node (v1) {} circle (.08);
\draw[color=black] (-2.5,2) node (v1) {} circle (.08);
\draw[color=black] (-2.5,1.5) node (v1) {} circle (.08);
\draw[color=black] (-2.5,1) node (v1) {} circle (.08);
\draw[color=black] (-1.5,1.5) node (v1) {} circle (.08); 
\draw[color=black] (-1.5,0.5) node (v1) {} circle (.08); 
\draw[color=black] (-1.5,1) node (v1) {} circle (.08);
\draw[color=black] (0,1.5) node (v1) {} circle (.08); 
\draw[color=black] (0,0.5) node (v1) {} circle (.08); 
\draw[color=black] (0,1) node (v1) {} circle (.08);
\draw[color=black] (-0.5,1.5) node (v1) {} circle (.08); 
\draw[color=black] (-0.5,0.5) node (v1) {} circle (.08); 
\draw[color=black] (-0.5,1) node (v1) {} circle (.08);
\draw[color=black] (-3.5,2.5) node (v1) {} circle (.08);
\draw[color=black] (-3,1.5) node (v1) {} circle (.08);
\draw[color=black] (-3,2) node (v1) {} circle (.08);
\draw[color=black] (-3,2.5) node (v1) {} circle (.08);
\draw[color=black] (-3.5,0.5) node (v1) {} circle (.08);
\draw[color=black] (-3.5,1) node (v1) {} circle (.08);
\draw[color=black] (-3.5,1.5) node (v1) {} circle (.08);
\draw[color=black] (-3.5,2) node (v1) {} circle (.08);
\draw[color=black] (-3,0.5) node (v1) {} circle (.08);
\draw[color=black] (-3,1) node (v1) {} circle (.08);
\draw[color=black] (-1,0.5) node (v1) {} circle (.08);

\node[scale=0.7] at (-5,0.25) {$0$}; 
\node[scale=0.7] at (-4.5,0.25) {$1$}; 
\node[scale=0.7] at (-4,0.25) {$2$}; 
\node[scale=0.7] at (-3.5,0.25) {$3$}; 
\node[scale=0.7] at (-3,0.25) {$4$}; 
\node[scale=0.7] at (-5.25,0.5) {$0$};
\node[scale=0.7] at (-5.25,1) {$1$};
\node[scale=0.7] at (-5.25,2) {$3$};
\node[scale=0.7] at (-2.5,0.25) {$5$};
\node[scale=0.7] at (-5.25,1.5) {$2$}; 
\node[scale=0.7] at (-5.25,2.5) {$4$}; 
\node[scale=0.7] at (0,0.25) {$10$};
\node[scale=0.7] at (-0.5,0.25) {$9$};
\node[scale=0.7] at (-1,0.25) {$8$};
\node[scale=0.7] at (-1.5,0.25) {$7$};
\node[scale=0.7] at (-2,0.25) {$6$};

\node [scale=0.8] at (-3.875,0.625) {a};
\node [scale=0.8] at (-4.375,1.875) {b};
\node [scale=0.8] at (-4.375,2.375) {c};

\end{tikzpicture}

%% file: figures/coalescence_ex2.tex
\begin{tikzpicture}[scale=1.3]
\definecolor{color0}{rgb}{0.12156862745098,0.466666666666667,0.705882352941177}

\draw[very thin]  (-5,2.5) rectangle (-2.5,0.5); 
\draw[very thin]  (-2.5,1.5) rectangle (0,0.5);
\draw [help lines,  step=0.5cm] (-5,0.5) node (v19) {} grid (-2.5,2.5); 
\draw [help lines,  step=0.5cm] (-2.5,0.5) node (v19) {} grid (0,1.5); 
\draw[fill, color=black] (-4.5,0.5) node (v1) {} circle (.08);
\draw[fill, color=black] (-4.5,2.5) node (v1) {} circle (.08);
\draw[fill, color=black] (-4.5,2) node (v1) {} circle (.08);
\draw[fill, color=black] (-4.5,1.5) node (v1) {} circle (.08);
\draw[fill, color=black] (-4.5,1) node (v1) {} circle (.08);
\draw[fill, color=black] (-5,0.5) node (v0) {} circle (.08); 
\draw[fill, color=black] (-5,2.5) node (v1) {} circle (.08);
\draw[fill, color=black] (-5,2) node (v1) {} circle (.08);
\draw[fill, color=black] (-5,1.5) node (v1) {} circle (.08);
\draw[fill, color=black] (-5,1) node (v1) {} circle (.08);
\draw[fill, color=black] (-4,0.5) node (v1) {} circle (.08);
\draw[fill, color=black] (-4,2.5) node (v1) {} circle (.08);
\draw[fill, color=black] (-4,2) node (v1) {} circle (.08);
\draw[fill, color=black] (-4,1.5) node (v1) {} circle (.08);
\draw[fill, color=black] (-4,1) node (v1) {} circle (.08);
\draw[fill, color=black] (-2.5,0.5) node (v1) {} circle (.08);
\draw[color=black] (-1,1.5) node (v1) {} circle (.08);
\draw[fill, color=black] (-2,0.5) node (v1) {} circle (.08);
\draw[fill, color=color0] (-2,1.5) node (v1) {} circle (.08);
\draw[fill, color=color0] (-2,1) node (v1) {} circle (.08);
\draw[color=black] (-1,1) node (v1) {} circle (.08);
\draw[fill, color=black] (-2.5,2.5) node (v1) {} circle (.08);
\draw[fill, color=black] (-2.5,2) node (v1) {} circle (.08);
\draw[fill, color=black] (-2.5,1.5) node (v1) {} circle (.08);
\draw[fill, color=black] (-2.5,1) node (v1) {} circle (.08);
\draw[color=black] (-1.5,1.5) node (v1) {} circle (.08); 
\draw[fill, color=color0] (-1.5,0.5) node (v1) {} circle (.08); 
\draw[color=black] (-1.5,1) node (v1) {} circle (.08);
\draw[color=black] (0,1.5) node (v1) {} circle (.08); 
\draw[color=black] (0,0.5) node (v1) {} circle (.08); 
\draw[color=black] (0,1) node (v1) {} circle (.08);
\draw[color=black] (-0.5,1.5) node (v1) {} circle (.08); 
\draw[color=black] (-0.5,0.5) node (v1) {} circle (.08); 
\draw[color=black] (-0.5,1) node (v1) {} circle (.08);
\draw[fill, color=black] (-3.5,2.5) node (v1) {} circle (.08);
\draw[fill, color=black] (-3,1.5) node (v1) {} circle (.08);
\draw[fill, color=black] (-3,2) node (v1) {} circle (.08);
\draw[fill, color=black] (-3,2.5) node (v1) {} circle (.08);
\draw[fill, color=black] (-3.5,0.5) node (v1) {} circle (.08);
\draw[fill, color=black] (-3.5,1) node (v1) {} circle (.08);
\draw[fill, color=black] (-3.5,1.5) node (v1) {} circle (.08);
\draw[fill, color=black] (-3.5,2) node (v1) {} circle (.08);
\draw[fill, color=black] (-3,0.5) node (v1) {} circle (.08);
\draw[fill, color=black] (-3,1) node (v1) {} circle (.08);
\draw[color=black] (-1,0.5) node (v1) {} circle (.08);

\node[scale=0.7] at (-5,0.25) {$0$}; 
\node[scale=0.7] at (-4.5,0.25) {$1$}; 
\node[scale=0.7] at (-4,0.25) {$2$}; 
\node[scale=0.7] at (-3.5,0.25) {$3$}; 
\node[scale=0.7] at (-3,0.25) {$4$}; 
\node[scale=0.7] at (-5.25,0.5) {$0$};
\node[scale=0.7] at (-5.25,1) {$1$};
\node[scale=0.7] at (-5.25,2) {$3$};
\node[scale=0.7] at (-2.5,0.25) {$5$};
\node[scale=0.7] at (-5.25,1.5) {$2$}; 
\node[scale=0.7] at (-5.25,2.5) {$4$}; 
\node[scale=0.7] at (0,0.25) {$10$};
\node[scale=0.7] at (-0.5,0.25) {$9$};
\node[scale=0.7] at (-1,0.25) {$8$};
\node[scale=0.7] at (-1.5,0.25) {$7$};
\node[scale=0.7] at (-2,0.25) {$6$};

\node[scale=0.8] at (-1.375,0.625) {a};
\node[scale=0.8] at (-1.875,1.625) {c};
\node[scale=0.8] at (-1.875,1.125) {b};
\end{tikzpicture}

%% file: figures/proof_fig1.tex
	\usetikzlibrary{decorations.pathreplacing}
	\begin{tikzpicture}[scale=1.3]
	
	\draw [help lines,  step=0.5cm] (-3.5, -3.5) node (v19) {} grid (-1,1.5) node (v2) {}; 
	\draw [help lines,  step=0.5cm] (-1,-3.5) node (v19) {} grid (1.5,-1) node (v2) {}; 
	
	\node[scale=0.9] at (-3.25,-3.75) {$0$};  
	\node[scale=0.9] at (-1.25,-3.75) {$K+T-1$};   
	\node[scale=0.9] at (1.25,-3.75) {$2K+2T-2$};

	\node[scale=0.9] at (-3.875,-3.25) {$0$};    
	  
	\node[scale=0.9] at (-3.875,-1.25) {$m-1$};
	\node[scale=0.9] at (-3.875,1.25) {$L-1$};  
	 
	\draw[fill, color=black] (-3.25,-3.25) circle (.08);  
	\draw[fill, color=black] (-3.25,1.25) circle (.08);  
	\draw[fill, color=black] (-3.25,0.75) circle (.08);

	\draw[fill, color=black] (-2.75,1.25) circle (.08);  
	\draw[fill, color=black] (-2.75,0.75) circle (.08);

	\draw[fill, color=black] (-2.75,-1.25) circle (.08);

	\draw[fill, color=black] (-2.75,-3.25) circle (.08);  
	\draw[fill, color=black] (-3.25,-1.25) circle (.08);

	\draw[fill, color=black] (-2.25,-1.25) circle (.08);

	\draw[fill, color=black] (-2.25,-3.25) circle (.08);  
	\draw[fill, color=black] (-2.25,-0.75) circle (.08);  
	\draw[fill, color=black] (-3.25,-0.75) circle (.08);  
	\draw[fill, color=black] (-2.75,-0.75) circle (.08);

	\draw[color=black] (-2.25,1.25) circle (.08);  
	\draw[color=black] (-2.25,1.25) circle (.08);  
	\draw[color=black] (-2.25,-0.25) circle (.08);  
	\draw[color=black] (-2.25,0.75) circle (.08);  
	
	\draw[color=black] (-1.75,-3.25) circle (.08);  
	\draw[color=black] (-1.75,-2.25) circle (.08);

	\node at (-3.25,-1.625) {$\vdots$};
	\node at (-3.25,-2.125) {$\vdots$};
	\node at (-3.25,0.375) {$\vdots$};
	\node at (-3.25,-0.125) {$\vdots$};
	\node at (-2.75,0.375) {$\vdots$};
	\node at (-2.75,-0.125) {$\vdots$};
	\node at (-2.25,0.375) {$\vdots$};
	\node at (-2.75,-1.625) {$\vdots$};
	\node at (-2.75,-2.625) {$\vdots$};
	\node at (-2.75,-2.125) {$\vdots$};
	\node at (-3.25,-2.625) {$\vdots$};
	\node at (-2.25,-2.625) {$\vdots$};
	\node at (-2.25,-2.125) {$\vdots$};
	\node at (-2.25,-1.625) {$\vdots$};

	\node at (-1.75,-2.625) {$\vdots$};
	
\node (v1) at (-2,1.5) {};
\node (v3) at (-2,-0.5) {};
\draw [decorate, decoration={brace, amplitude=5pt}] (v1) -- (v3);
\node at (-1.75,0.5) {$\mu$};

\node (v4) at (-1.5,-2) {};
\node (v5) at (-1.5,-3.5) {};
\draw [decorate, decoration={brace, amplitude=5pt}] (v4) -- (v5);
\node at (-1.25,-2.75) {$\xi$};

\node at (-2.25,-3.75) {$p_x$};
\node at (-3.875,-0.75) {$p_y$};
	\node at (-3.875,-0.875) {$\vdots$};
	
	\end{tikzpicture}

%% file: figures/proof_fig2.tex
\usetikzlibrary{decorations.pathreplacing}
\begin{tikzpicture}[scale=1.3]

\draw [help lines,  step=0.5cm] (-3.5, -3.5) node (v19) {} grid (-1,1.5) node (v2) {}; 
\draw [help lines,  step=0.5cm] (-1,-3.5) node (v19) {} grid (1.5,-1) node (v2) {}; 

\node[scale=0.9] at (-3.25,-3.75) {$0$};  
\node[scale=0.9] at (-1.25,-3.75) {$K+T-1$};   
\node[scale=0.9] at (1.25,-3.75) {$2K+2T-2$};

\node[scale=0.9] at (-3.875,-3.25) {$0$};    
  
\node[scale=0.9] at (-3.875,-1.25) {$m-1$};
\node[scale=0.9] at (-3.875,1.25) {$L-1$};

\draw[color=black] (-2.25,1.25) circle (.08);  
\draw[color=black] (-2.25,1.25) circle (.08);  
\draw[color=black] (-2.25,-0.25) circle (.08);  
\draw[color=black] (-2.25,0.75) circle (.08);  

\draw[color=black] (-1.75,-3.25) circle (.08);  
\draw[color=black] (-1.75,-2.25) circle (.08);

\node at (-2.25,0.375) {$\vdots$};
\node at (-1.75,-2.625) {$\vdots$};
\node at (-3.25,-2.625) {$\vdots$};
\node at (-3.25,-2.125) {$\vdots$};

\node at (-3.25,-3.25) {$\Phi_1$};

\node at (-3.25,-1.75) {$\Phi_{v_y+1}$};

\node (v1) at (-2,1.5) {};
\node (v3) at (-2,-0.5) {};
\draw [decorate, decoration={brace, amplitude=5pt}] (v1) -- (v3);
\node at (-1.75,0.5) {$\mu$};

\node (v4) at (-1.5,-2) {};
\node (v5) at (-1.5,-3.5) {};
\draw [decorate, decoration={brace, amplitude=5pt}] (v4) -- (v5);
\node at (-1.25,-2.75) {$\xi$};

\end{tikzpicture}

%% file: figures/proof_fig3.tex
\usetikzlibrary{decorations.pathreplacing}
\begin{tikzpicture}[scale=1.3]

\draw [help lines,  step=0.5cm] (-3.5, -3.5) node (v19) {} grid (-1,1.5) node (v2) {}; 
\draw [help lines,  step=0.5cm] (-1,-3.5) node (v19) {} grid (1.5,-1) node (v2) {}; 

\node[scale=0.9] at (-3.25,-3.75) {$0$};  
\node[scale=0.9] at (-1.25,-3.75) {$K+T-1$};   
\node[scale=0.9] at (1.25,-3.75) {$2K+2T-2$};

\node[scale=0.9] at (-3.875,-3.25) {$0$};    
  
\node[scale=0.9] at (-3.875,-1.25) {$m-1$};
\node[scale=0.9] at (-4,1.25) {$L-1$};

\draw[color=black] (-2.25,1.25) circle (.08);  
\draw[color=black] (-2.25,1.25) circle (.08);  
\draw[color=black] (-2.25,-0.25) circle (.08);  
\draw[color=black] (-2.25,0.75) circle (.08);  

\draw[color=black] (-1.75,-3.25) circle (.08);  
\draw[color=black] (-1.75,-2.25) circle (.08);

\node at (-2.25,0.375) {$\vdots$};
\node at (-1.75,-2.625) {$\vdots$};
\node at (-3.25,-2.625) {$\vdots$};
\node at (-3.25,0.875) {$\vdots$};
\node at (-3.25,0.375) {$\vdots$};

\node at (-3.25,-3.25) {$\Phi_1$};
\node at (-3.25,-2.25) {$\Phi_{\xi}$};

\node at (-3.25,1.25) {$\Phi_{v_y+1}$};
\node at (-3.25,-0.25) {$\Phi_{\xi+1}$};

\node (v1) at (-2,1.5) {};
\node (v3) at (-2,-0.5) {};
\draw [decorate, decoration={brace, amplitude=5pt}] (v1) -- (v3);
\node at (-1.75,0.5) {$\mu$};

\node (v4) at (-1.5,-2) {};
\node (v5) at (-1.5,-3.5) {};
\draw [decorate, decoration={brace, amplitude=5pt}] (v4) -- (v5);
\node at (-1.25,-2.75) {$\xi$};

\end{tikzpicture}

%% file: main.bbl
\begin{thebibliography}{10}
\providecommand{\url}[1]{#1}
\csname url@samestyle\endcsname
\providecommand{\newblock}{\relax}
\providecommand{\bibinfo}[2]{#2}
\providecommand{\BIBentrySTDinterwordspacing}{\spaceskip=0pt\relax}
\providecommand{\BIBentryALTinterwordstretchfactor}{4}
\providecommand{\BIBentryALTinterwordspacing}{\spaceskip=\fontdimen2\font plus
\BIBentryALTinterwordstretchfactor\fontdimen3\font minus
  \fontdimen4\font\relax}
\providecommand{\BIBforeignlanguage}[2]{{%
\expandafter\ifx\csname l@#1\endcsname\relax
\typeout{** WARNING: IEEEtran.bst: No hyphenation pattern has been}%
\typeout{** loaded for the language `#1'. Using the pattern for}%
\typeout{** the default language instead.}%
\else
\language=\csname l@#1\endcsname
\fi
#2}}
\providecommand{\BIBdecl}{\relax}
\BIBdecl

\bibitem{hasircioglu2021speeding}
B.~Hasircioglu, J.~G{\'o}mez-Vilardeb{\'o}, and D.~Gunduz, ``Speeding up
  private distributed matrix multiplication via bivariate polynomial codes,''
  \emph{2021 IEEE International Symposium on Information Theory (ISIT)}, pp.
  1853--1858, 2021.

\bibitem{lee2017speeding}
K.~Lee, M.~Lam, R.~Pedarsani, D.~Papailiopoulos, and K.~Ramchandran, ``Speeding
  up distributed machine learning using codes,'' \emph{IEEE Transactions on
  Information Theory}, vol.~64, no.~3, pp. 1514--1529, 2017.

\bibitem{yu_polynomial_2017}
Q.~Yu, M.~Maddah-Ali, and S.~Avestimehr, ``Polynomial codes: an optimal design
  for high-dimensional coded matrix multiplication,'' in \emph{Advances in
  Neural Information Processing Systems}, 2017, pp. 4403--4413.

\bibitem{dutta_optimal_2019}
S.~Dutta, M.~Fahim, F.~Haddadpour, H.~Jeong, V.~Cadambe, and P.~Grover, ``On
  the optimal recovery threshold of coded matrix multiplication,'' \emph{IEEE
  Transactions on Information Theory}, vol.~66, no.~1, pp. 278--301, 2019.

\bibitem{yu_straggler_2018-1}
Q.~Yu, M.~A. Maddah-Ali, and A.~S. Avestimehr, ``Straggler mitigation in
  distributed matrix multiplication: Fundamental limits and optimal coding,''
  \emph{IEEE Transactions on Information Theory}, vol.~66, pp. 1920--1933,
  2020.

\bibitem{jia2019cross}
Z.~Jia and S.~A. Jafar, ``Cross subspace alignment codes for coded distributed
  batch computation,'' \emph{IEEE Transactions on Information Theory}, vol.~67,
  pp. 2821--2846, 2021.

\bibitem{kiani_exploitation_2018}
S.~Kiani, N.~Ferdinand, and S.~C. Draper, ``Exploitation of stragglers in coded
  computation,'' in \emph{IEEE International Symposium on Information Theory},
  2018.

\bibitem{amiri_computation_2018}
M.~M. Amiri and D.~G{\"u}nd{\"u}z, ``Computation scheduling for distributed
  machine learning with straggling workers,'' \emph{IEEE Transactions on Signal
  Processing}, vol.~67, no.~24, pp. 6270--6284, 2019.

\bibitem{ozfatura2020straggler}
E.~Ozfatura, S.~Ulukus, and D.~G{\"u}nd{\"u}z, ``Straggler-aware distributed
  learning: Communication--computation latency trade-off,'' \emph{Entropy},
  vol.~22, no.~5, p. 544, 2020.

\bibitem{hasircioglu2020bivariate}
B.~Has{\i}rc{\i}o{\u{g}}lu, J.~G{\'o}mez-Vilardeb{\'o}, and D.~G{\"u}nd{\"u}z,
  ``Bivariate polynomial coding for efficient distributed matrix
  multiplication,'' \emph{IEEE Journal on Selected Areas in Information
  Theory}, vol.~2, no.~3, pp. 814--829, 2021.

\bibitem{chang2018capacity}
W.-T. Chang and R.~Tandon, ``On the capacity of secure distributed matrix
  multiplication,'' in \emph{2018 IEEE Global Communications Conference
  (GLOBECOM)}.\hskip 1em plus 0.5em minus 0.4em\relax IEEE, 2018, pp. 1--6.

\bibitem{kakar2018rate}
J.~Kakar, S.~Ebadifar, and A.~Sezgin, ``Rate-efficiency and
  straggler-robustness through partition in distributed two-sided secure matrix
  computation,'' \emph{arXiv preprint arXiv:1810.13006}, 2018.

\bibitem{d2020gasp}
R.~G.~L. {D’Oliveira}, S.~{El Rouayheb}, and D.~{Karpuk}, ``{GASP} codes for
  secure distributed matrix multiplication,'' \emph{IEEE Transactions on
  Information Theory}, vol.~66, no.~7, pp. 4038--4050, 2020.

\bibitem{aliasgari2020private}
M.~Aliasgari, O.~Simeone, and J.~Kliewer, ``Private and secure distributed
  matrix multiplication with flexible communication load,'' \emph{IEEE
  Transactions on Information Forensics and Security}, vol.~15, pp. 2722--2734,
  2020.

\bibitem{jia2021capacity}
Z.~Jia and S.~A. Jafar, ``On the capacity of secure distributed batch matrix
  multiplication,'' \emph{IEEE Transactions on Information Theory}, vol.~67,
  no.~11, pp. 7420--7437, 2021.

\bibitem{kakar2019uplink}
J.~Kakar, A.~Khristoforov, S.~Ebadifar, and A.~Sezgin, ``Uplink cost adjustable
  schemes in secure distributed matrix multiplication,'' \emph{2020 IEEE
  International Symposium on Information Theory (ISIT)}, pp. 1124--1129, 2020.

\bibitem{mital2020secure}
N.~Mital, C.~Ling, and D.~Gunduz, ``Secure distributed matrix computation with
  discrete {F}ourier transform,'' \emph{arXiv preprint arXiv:2007.03972}, 2020.

\bibitem{bitar2020rateless}
R.~Bitar, M.~Xhemrishi, and A.~Wachter-Zeh, ``Rateless codes for private
  distributed matrix-matrix multiplication,'' \emph{arXiv preprint
  arXiv:2004.12925}, 2020.

\bibitem{liang2014tofec}
G.~Liang and U.~C. Kozat, ``{TOFEC}: Achieving optimal throughput-delay
  trade-off of cloud storage using erasure codes,'' in \emph{IEEE INFOCOM
  2014-IEEE Conference on Computer Communications}.\hskip 1em plus 0.5em minus
  0.4em\relax IEEE, 2014, pp. 826--834.

\bibitem{hoffmanlinear}
K.~Hoffman and R.~Kunze, ``Linear algebra,'' \emph{Englewood Cliffs, New
  Jersey}, 1971.

\bibitem{felix_fontein_2009}
\BIBentryALTinterwordspacing
F.~Fontein, ``The {H}asse derivative,'' Aug 2009. [Online]. Available:
  \url{https://math.fontein.de/2009/08/12/the-hasse-derivative/}
\BIBentrySTDinterwordspacing

\bibitem{schwartz1980fast}
J.~T. Schwartz, ``Fast probabilistic algorithms for verification of polynomial
  identities,'' \emph{Journal of the ACM (JACM)}, vol.~27, no.~4, pp. 701--717,
  1980.

\bibitem{liesen_linear_2015}
J.~Liesen and V.~Mehrmann, \emph{Linear algebra}, ser. Springer Undergraduate
  Mathematics Series.

\end{thebibliography}
